\documentclass[letterpaper]{article} % DO NOT CHANGE THIS
\usepackage{aaai23}
\usepackage{times}  % DO NOT CHANGE THIS
\usepackage{helvet}  % DO NOT CHANGE THIS
\usepackage{courier}  % DO NOT CHANGE THIS
\usepackage[hyphens]{url}  % DO NOT CHANGE THIS
\usepackage{graphicx} % DO NOT CHANGE THIS
\usepackage{natbib}  % DO NOT CHANGE THIS AND DO NOT ADD ANY OPTIONS TO IT
\usepackage{caption} % DO NOT CHANGE THIS AND DO NOT ADD ANY OPTIONS TO IT
\usepackage{algorithm}
\usepackage{algorithmic}

\usepackage{newfloat}
\usepackage{listings}
\DeclareCaptionStyle{ruled}{labelfont=normalfont,labelsep=colon,strut=off} % DO NOT CHANGE THIS
\floatstyle{ruled}
\newfloat{listing}{tb}{lst}{}
\floatname{listing}{Listing}
\usepackage{subfigure}
\usepackage{tikz}
\usepackage{pgfplots}
\usetikzlibrary{pgfplots.groupplots}
\pgfplotsset{compat=1.3}

\usepackage{multirow}

\usepackage{array}
\usepackage{makecell}
\newcolumntype{L}[1]{>{\raggedright\let\newline\\\arraybackslash\hspace{0pt}}m{#1}}
\newcolumntype{C}[1]{>{\centering\let\newline\\\arraybackslash\hspace{0pt}}m{#1}}
\newcolumntype{R}[1]{>{\raggedleft\let\newline\\\arraybackslash\hspace{0pt}}m{#1}}
\usepackage{booktabs}

\newcommand{\reals}{\mathbb{R}}
\newcommand{\E}{\mathbb{E}}

\newcommand{\tr}{\text{Tr}}
\newcommand{\diag}{\text{diag}}
\newcommand{\dm}{\textup{\text{dim}}}
\newcommand{\sn}{\textup{\text{span}}}
\newcommand{\cl}{\textup{\text{col}}}

\newenvironment{customlegend}[1][]{%
    \begingroup
    \csname pgfplots@init@cleared@structures\endcsname
    \pgfplotsset{#1}%
}{%
    \csname pgfplots@createlegend\endcsname
    \endgroup
}%

\def\addlegendimage{\csname pgfplots@addlegendimage\endcsname}

\usepackage{amsmath}
\usepackage{amssymb}
\usepackage{mathtools}
\usepackage{amsthm}

\usepackage{amsfonts}
\usepackage{bm}
\usepackage{bbm}

\theoremstyle{plain}
\newtheorem{theorem}{Theorem}

\newtheorem{corollary}[theorem]{Corollary}
\theoremstyle{definition}

\theoremstyle{remark}

\title{Task-Aware Network Coding Over Butterfly Network}
\author{
    Jiangnan Cheng\textsuperscript{\rm 1}, Sandeep Chinchali\textsuperscript{\rm 2}, Ao Tang\textsuperscript{\rm 1}
}
\affiliations{
    \textsuperscript{\rm 1}School of Electrical and Computer Engineering, Cornell University, Ithaca, NY\\
    \textsuperscript{\rm 2}Department of Electrical and Computer Engineering, The University of Texas at Austin, Austin, TX
}

\usepackage{bibentry}
\begin{document}

\maketitle

\begin{abstract}
Network coding allows distributed information sources such as sensors to efficiently compress and transmit data to distributed receivers across a bandwidth-limited network. Classical network coding is largely task-agnostic -- the coding schemes mainly aim to faithfully reconstruct data at the receivers, regardless of what ultimate task the received data is used for. In this paper, we analyze a new \textit{task-driven} network coding problem, where distributed receivers pass transmitted data through machine learning (ML) tasks, which provides an opportunity to improve efficiency by 
transmitting salient task-relevant data representations. 
Specifically, we formulate a \textit{task-aware} network coding problem over a butterfly network in real-coordinate space, where lossy analog compression through principal component analysis (PCA) can be applied. A lower bound for the total loss function for the formulated problem is given, and necessary and sufficient conditions for achieving this lower bound are also provided. We introduce ML algorithms to solve the problem in the general case, and our evaluation demonstrates the effectiveness of task-aware network coding.
\end{abstract}

\section{Introduction}\label{sec:intro}

Distributed sensors measure rich sensory data which potentially are consumed by multiple distributed data receivers. On the other hand, network bandwidths remain limited and expensive, especially for wireless networks. For example, low Earth orbit satellites collect high-resolution Earth imagery, whose size goes up to few terabytes per day and is sent to geographically distributed ground stations, while in the best case one ground station can only download 80 GB from one satellite in a single pass \citep{vasisht2021l2d2}. Therefore, one is motivated to make efficient use of the existing network bandwidths for \textit{distributed} data sources and receivers.

Network coding \citep{ahlswede2000network} is an important technology which aims at maximizing the network throughput for multi-source multicasting with limited network bandwidths. Classical network coding literatures \citep{li2003linear,koetter2003algebraic,dougherty2005insufficiency,jaggi2005polynomial,ho2006random,chen2008utility} consider a pure network information flow problem from the information-theoretic view, where the demands for all the data receivers, either homogeneous or heterogeneous, are specified and the objective is to satisfy each demand with a rate (i.e., mutual information between the demand and the received data) as high as possible. However, in reality each data receiver may apply the received data to a different task, such as inference, perception and control, where different lossy data representations, even with the same rate, can produce totally different task losses. Hence it is highly prominent to transmit \textit{salient} task-relevant data representations to distributed receivers that satisfy the network topology and bandwidth constraints, rather than representations with the highest rate.

\begin{figure*}[t]
\centering

\subfigure[]{
\includegraphics[scale=0.7]{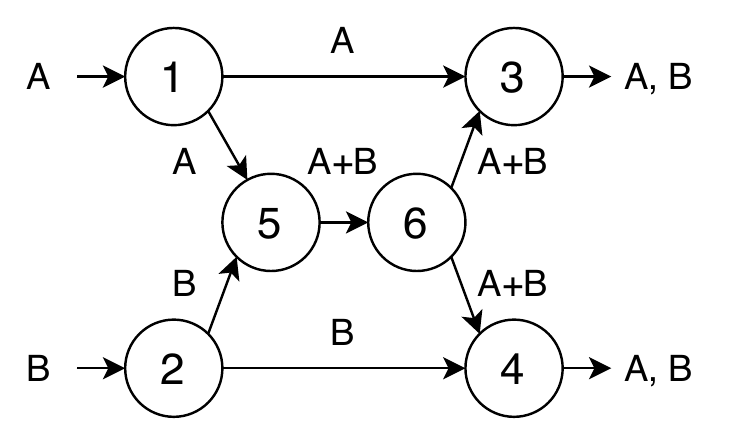}
\label{fig:network_coding}
}
\subfigure[]{
\includegraphics[scale=0.7]{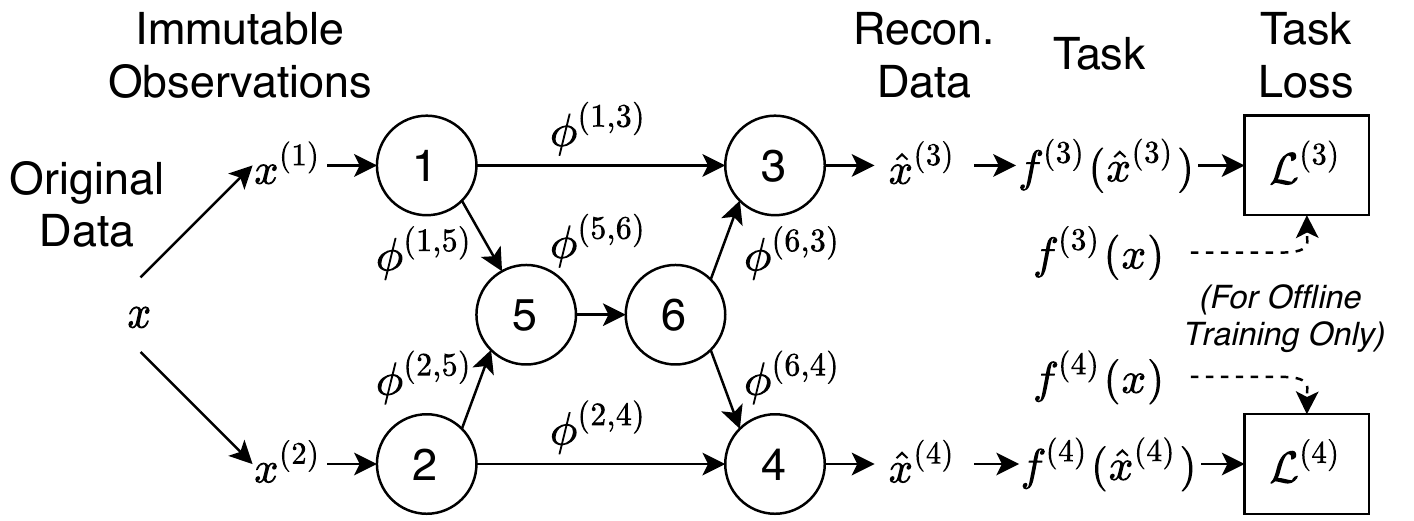}
\label{fig:architecture}
}

\caption{Network coding over butterfly network. Left (classical setting): Task-agnostic network coding in finite field. Node 3 can decode B through A+(A+B) where `+' represents exclusive or logic. 
% Node 4 can decode A similarly. 
Right (our setting): Task-aware network coding in real-coordinate space. Salient task-relevant data representations are transmitted to make efficient use of network bandwidths.}
\label{fig:two_network_codings}
\end{figure*}

Therefore, we formulate a concrete task-aware network coding problem in this paper --  \textit{task-aware} linear network coding over butterfly network, as shown in Fig. \ref{fig:architecture}. Butterfly network is a representative topology in many existing network coding literatures 
\citep{avestimehr2009approximate,parag2010queueing,soeda2011quantum}, 
and hence it suffices to demonstrate the benefit of making network coding task-aware. Moreover, the domain of our problem is multi-dimensional real-coordinate space $\reals^n$ rather than finite field $\text{GF}(\cdot)$ as in classical network coding literatures, which enables us to consider lossy analog compression (similar to \citet{wu2010renyi}) through principal component analysis (PCA) \citep{dunteman1989principal} rather than information-theoretic discrete compression.

\textbf{Related work}. Our work is broadly related to network coding and task-aware representation learning. First, beyond those classical network coding literatures, the two closest works to ours are \citet{liu2020neural} and \citet{whang2021neural}, where \textit{data-driven} approach is adopted in the general network coding and distributed source coding settings respectively, to determine a coding scheme that minimizes \textit{task-agnostic} reconstruction loss. In stark contrast, we aim at finding a linear network coding scheme that minimizes an overall \textit{task-aware} loss which incorporates heterogeneous task objectives of different receivers, and we show that in some cases such linear coding scheme can even be determined analytically. Second, our work is also related to network functional compression problem \citep{doshi2010functional,feizi2014network,shannon1956zero,slepian1973noiseless,ahlswede1975source,wyner1976rate,korner1979encode}, where a general function with distributed inputs over finite space is compressed. There's a similar task-aware loss function in our work, yet it corresponds to machine learning tasks over multi-dimensional real-coordinate space.
% yet it incorporates the heterogeneous task objectives of \textit{distributed} receivers, rather than a single receiver. 
Lastly, there have been a variety of works \citep{blau2019rethinking,nakanoya2021co,dubois2021lossy,zhang2021universal,cheng2021data} focusing on task-aware data compression for inference, perception and control tasks under a single-source single-destination setting which is similar to Shannon's rate-distortion theory \citep{shannon1959coding}, while in contrast we consider task-aware data compression in a \textit{distributed} setting.

\textbf{Contributions}. In light of prior work, our contributions are three-fold as follows. First, we formulate a task-aware network coding problem over butterfly network in real-coordinate space where lossy analog compression through PCA can be applied. Second, we give a lower bound for the formulated problem, and provide necessary condition and sufficient conditions for achieving such lower bound. Third, we adopt standard gradient descent algorithm to solve the formulated problem in the general case, and validate the effectiveness of task-aware network coding in our evaluation.

% \textbf{Organization}. In Section \ref{sec:preliminaries} we introduce network coding with a classical example and task-aware PCA. In Section \ref{sec:problem_formulation} we formulate task-aware network coding problem over butterfly network in real-coordinate space. In Section \ref{sec:analysis} we provide a theoretical lower bound for the considered problem as well as necessary condition and sufficient conditions when such lower bound is achievable. In Section \ref{sec:algorithm} we give the machine learning algorithm to solve the considered problem in the general case. Evaluation results are shown in Section \ref{sec:eval}, and Section \ref{sec:conclusion} concludes the paper.

\section{Preliminaries}\label{sec:preliminaries}
\subsection{Network Coding with a Classical Example}\label{subsec:preliminary_network_coding}

% \textbf{Network Coding with a Classical Example.} 
Network coding \cite{ahlswede2000network} is a technique to increase the network throughput for multi-source multicasting under limited network bandwidths. The key idea of network coding is to allow each node within the network to encode and decode data rather than simply routing it. 

A classical example over butterfly network in finite field $\text{GF}(2)$, as shown in Fig. \ref{fig:network_coding}, is widely used to illustrate the benefit of network coding. The butterfly network can be represented by a directed graph $\mathcal{G} = (\mathcal{V}, \mathcal{E})$. Here $\mathcal{V} = \{1, 2, \cdots, 6\}$ is the set of nodes, and $\mathcal{E} = \{(1, 3), (1, 5), (2, 4), (2, 5), (5, 6), (6, 3), (6, 4)\}$ is the set of edges, where $(i, j)$ represents an edge with source node $i$ and destination node $j$. Suppose each edge in $\mathcal{E}$ can only carry a single bit, and node 1 and 2 each have a single bit of information, denoted by A and B respectively, which are supposed to be multicast to both node 3 and 4. In this case, network coding, as illustrated in Fig. \ref{fig:network_coding}, makes such multicasting possible while routing cannot. The key idea is to encode A and B as A+B at node 5, where `+' here represents exclusive or logic. Node 3 is able to decode B through A+(A+B), and node 4 can decode A similarly.

% \subsection{Task-aware Principle Component Analysis (PCA)}\label{subsec:preliminary_pca}
\subsection{Task-aware PCA}\label{subsec:preliminary_pca}
% \textbf{Task-aware PCA.} 
PCA is a widely-used dimensionality-reduction technique, and has been used in \citet{nakanoya2021co,cheng2021data}, etc., for task-aware data compression under a single-source single-destination setting. 

Suppose we have an $n$-dimensional random vector $x \in \reals^n$, with mean $\E_x [x] = \bm{0}$ and positive definite covariance matrix $\Psi = \E_x[xx^\top]$ (i.e., $\text{rank}(\Psi) = n$). Consider the following task-aware data compression problem:
\begin{align}
\min_{D, E} \quad & \mathcal{L} = \E_x [\|f(x) - f(\hat{x}) \|_2^2]\\
\text{s.t.} \quad & \hat{x} = D E x, D \in \reals^{n \times Z}, E \in \reals^{Z \times n}
\end{align}
where $\hat{x}$ is the reconstructed vector through a bottlenecked channel which only transmits a low-dimensional vector in $\reals^Z$ such that $Z \leq n$, and $E \in \reals^{Z \times n}$ and $D \in \reals^{n \times Z}$ are the corresponding encoding and decoding matrices respectively. Loss function $\mathcal{L}$ is associated with a task function $f(\cdot) \in \reals^m$ and captures the mean-squared error between $f(x)$ and $f(\hat{x})$. In this paper we consider linear task function $f(x) = Kx$, where $K \in \reals^{m \times n}$ is called task matrix.
% We are allowed to encode $x$ to a low-dimensional signal $\phi \in \reals^Z$ through linear mapping $\phi = E x$, where $Z \leq n$ and $E \in \reals^{Z \times n}$. $\phi$, instead of $x$, is transmitted over a bottlenecked channel and we decode a reconstructed  $\hat{x}$ from $\phi$ through another linear mapping $\hat{x} = D \phi$, where $D \in \reals^{n \times Z}$. $\hat{x}$ is used as the input for an inference or control task, and corresponding loss function has the form of $\mathcal{L} = \E_x [\|K (x - \hat{x}) \|_2^2]$, where $K \in \reals^{m \times n}$ is referred to as task matrix. 

According to PCA, the optimal task loss $\mathcal{L}^*$ can be determined as follows. Suppose the Cholesky decomposition of $\Psi$ is $\Psi = LL^\top$ where $L \in \reals^{n \times n}$ is a lower triangular matrix with positive diagonal entries, and the eigen-values in descending order and the corresponding normalized eigen-vectors of Gram matrix $S = L^\top K^\top K L$ are $\mu_{1}, \mu_{2}, \cdots, \mu_{n}$ and $u_{1}, u_{2}, \cdots, u_{n}$ respectively. Then we have $\mathcal{L}^* = \sum_{i=Z+1}^n \mu_i$, and if the eigen-gap $\mu_{Z} - \mu_{Z+1} > 0$ (define $\mu_{n+1} = 0$), we must have $\cl(E^\top) = \sn(\{L^{-\top} u_{1}, L^{-\top} u_{2}, \cdots, L^{-\top} u_{Z}\})$ to achieve minimum task loss, where $\cl(\cdot)$ denotes the column space of a matrix 
% (i.e., linear span of the column vectors of a matrix) 
and $\sn(\cdot)$ denotes the linear span of a set of vectors. See appendix for a detailed derivation.
 % we then let the column vectors of $E^\top$ be $L^{-\top} u_{1}, L^{-\top} u_{2}, \cdots, L^{-\top} u_{Z}$, and let the column vectors of $D$ be $L u_{1}, L u_{2}, \cdots, L u_{Z}$.

\section{Problem Formulation}\label{sec:problem_formulation}

We now formulate a task-aware network coding problem over butterfly network, as shown in Fig. \ref{fig:architecture}. The key differences between our formulation and the classical example in Fig. \ref{fig:network_coding} are: 1) our formulation has a heterogeneous task objective for each receiver while the classical example does not; 2) the domain of code is multi-dimensional real-coordinate space in our formulation rather than finite space as in the classical example, and hence PCA can be applied.

\textbf{Data}. The original data is a random vector $x = [x_1, x_2, \cdots, x_n]^\top \in \reals^n$, where $x_i \in \reals$ is a random variable, $\forall i \in \{1, 2, \cdots, n\}$. Without loss of generality, we assume $\E_x [x] = \bm{0}$, or else we replace $x$ by $x - \E_x [x]$. We also let $\Psi= \E_x[xx^\top]$ be the covariance matrix of $x$.

\textbf{Data observations}. Node 1 and 2 have immutable partial observations of $x$, denoted by $x^{(1)} \in \reals^a$ and $x^{(2)} \in \reals^b$, respectively. Here observations $x^{(1)}$ and $x^{(2)}$ are composed of $a$ and $b$ different dimensions of $x$, respectively; and each $x_i$ exists in \textit{at least} one of the two observations. Therefore, we have $\max\{a, b\} \leq n \leq a+b$. Without loss of generality, we let $x^{(1)} = [x_1, x_2, \cdots, x_a]^\top$ and $x^{(2)} = [x_{n-b+1}, x_{n-b+2}, \cdots, x_n]^\top$.
% , as shown in Fig. \ref{fig:obs}. 
That is, $x_{1:n-b}$ and $x_{a+1:n}$ are node 1's and node 2's exclusive observations respectively, and $x_{n-b+1:a}$ are their mutual observations.

\textbf{Data transmission}. We assume all the edges have the same capacity $Z$, which represents the number of dimensions in real-coordinate space here. And $\forall (i, j) \in \mathcal{E}$, we use $\phi^{(i,j)} \in \reals^Z$ to denote the random vector that transmits over the edge $(i, j)$.
% Since we consider a bottlenecked butterfly network, we assume that $Z \leq \min\{a, b\}$. 
Notice that for each edge $(i, j) \in \mathcal{E}' = \{(1,3), (1,5), (2,4), (2,5), (5,6)\}$, the overall number of input dimensions for node $i$ can be larger than $Z$, so we use linear mappings to transform the input signal to a low-dimensional signal in $\reals^Z$: 
\begin{align}
& \phi^{(1,3)} = E^{(1, 3)} x^{(1)}, \phi^{(1,5)} = E^{(1, 5)} x^{(1)},\label{eq:encoding_1}\\
& \phi^{(2,4)} = E^{(2, 4)} x^{(2)}, \phi^{(2,5)} = E^{(2, 5)} x^{(2)},\\
& \phi^{(5,6)} = E^{(5, 6)}
\begin{bmatrix}
\phi^{(1,5)}\\
\phi^{(2,5)}
\end{bmatrix},
\end{align}
where $E^{(1, 3)}, E^{(1, 5)} \in \reals^{Z \times a}$, $E^{(2, 4)}, E^{(2, 5)} \in \reals^{Z \times b}$, and $E^{(5, 6)} \in \reals^{Z \times 2Z}$ are encoding matrices. Node 6 simply multicasts the data received from node 5 to node 3 and 4, i.e., $\phi^{(6,3)} = \phi^{(6,4)} = \phi^{(5,6)}$.

\textbf{Data reconstructions}. Node 3 and 4 aim to reconstruct the original data $x$, through the aggregated inputs they received from their respective input edges. The corresponding decoder functions are:
\begin{align}
& \hat{x}^{(3)} = D^{(3)}
\begin{bmatrix}
\phi^{(1,3)}\\
\phi^{(6,3)}
\end{bmatrix},
\hat{x}^{(4)} = D^{(4)}
\begin{bmatrix}
\phi^{(2,4)}\\
\phi^{(6,4)}
\end{bmatrix},\label{eq:decoding_2}
\end{align}
where $D^{(3)}, D^{(4)} \in \reals^{n \times 2Z}$ are decoding matrices for node 3 and node 4 respectively, and $\hat{x}^{(3)}$ and $\hat{x}^{(4)}$ are the reconstructed data at node $3$ and node $4$ respectively.

\textbf{Task objectives}. Node $i$ ($\forall i \in \{3, 4\}$) uses the reconstructed data $\hat{x}^{(i)}$ as the input for a task with the following loss function:
\begin{align}
\mathcal{L}^{(i)} = \E_x [\|f^{(i)} (x) - f^{(i)}(\hat{x}^{(i)}) \|_2^2], \quad \forall i \in \{3, 4\}
\end{align}
where $f^{(i)}(x) = K^{(i)} x$ with task matrix $K^{(i)} \in \reals^{m_i \times n}$.
% \begin{align}
% \mathcal{L}^{(i)} = \E_x [\|K^{(i)} (x - \hat{x}^{(i)}) \|_2^2], \quad \forall i \in \{3, 4\}
% \end{align}
% where $K^{(i)} \in \reals^{m_i \times n}$ is a task matrix.
Our overall task loss $\mathcal{L}_\text{total}$ is the sum of $\mathcal{L}^{(3)}$ and $\mathcal{L}^{(4)}$:
\begin{align}
\mathcal{L}_\text{total} = \mathcal{L}^{(3)} + \mathcal{L}^{(4)}.
\end{align}

\textbf{Task-aware network coding problem}. The problem can be written as an optimization problem:
\begin{align}
\min_{E^{(i, j)}, D^{(i)}} \quad & \mathcal{L}_\text{total}, \quad \text{s.t.} \quad \text{Eq.}(\ref{eq:encoding_1})-(\ref{eq:decoding_2})\label{eq:prob}
\end{align}
where we find the optimal encoder and decoder parameters to minimize the overall task loss $\mathcal{L}_\text{total}$. And we denote the problem by $TaskAwareCoding(n, \Psi, a, b, Z, K^{(3)}, K^{(4)})$ for given parameters $n, \Psi, a, b, Z, K^{(3)}, K^{(4)}$.

\section{Analysis}\label{sec:analysis}
In this section, we give detailed analysis towards the task-aware network coding problem. We first provide a lower bound $\mathcal{L}_{\text{total,lb}}$ which may not be always achievable, and then discuss necessary condition and sufficient conditions for $\mathcal{L}^*_{\text{total}} = \mathcal{L}_{\text{total,lb}}$.

% In this section, we give detailed analysis towards the task-aware network coding problem. In Sec. \ref{subsec:analysis_lb} we provide a lower bound $\mathcal{L}_{\text{total,lb}}$ which may not be always achievable, and in Sec. \ref{subsec:analysis_nc} and Sec. \ref{subsec:analysis_sc} we discuss necessary condition and sufficient conditions for $\mathcal{L}^*_{\text{total}} = \mathcal{L}_{\text{total,lb}}$, respectively.

\subsection{Lower bound $\mathcal{L}_{\text{total,lb}}$} \label{subsec:analysis_lb}

We first show in the following Theorem \ref{thm:transformation} that making the assumption of $\text{rank}(\Psi) = n$
% , i.e., $\Psi$ is positive-semidefinite, 
doesn't make the the task-aware network coding problem lose generality. 
% The proof is given in the appendix.
% due to space limit. 

\begin{theorem}\label{thm:transformation}
For any set of parameters $n, \Psi, a, b, Z, K^{(3)}, K^{(4)}$, we can transform $TaskAwareCoding(n, \Psi, a, b, Z, K^{(3)}, K^{(4)})$ to $TaskAwareCoding(\tilde{n}, \tilde{\Psi}, \tilde{a}, \tilde{b}, Z, \tilde{K}^{(3)}, \tilde{K}^{(4)})$ where $\tilde{n}, \tilde{\Psi}, \tilde{a}, \tilde{b}, Z, \tilde{K}^{(3)}, \tilde{K}^{(4)}$ is a set of parameters with $\text{rank}(\tilde{\Psi}) = \tilde{n}$, such that their optimal overall task losses are equal, and an optimal solution for one problem can be transformed to the optimal solution for another linearly.
\end{theorem}

\begin{proof}
Assume the top-$\tilde{n}$ eigen-values of $\Psi$ are greater than zero, where $\tilde{n} \leq n$. We use $\lambda_1, \cdots, \lambda_{\tilde{n}}$ to denote these eigen-values and $q_1, \cdots, q_{\tilde{n}}$ to denote the corresponding normalized eigen-vectors. Moreover, we let $\Lambda = \diag(\lambda_1, \cdots, \lambda_{\tilde{n}}) \in \reals^{\tilde{n} \times \tilde{n}}$ and $Q = [q_1, \cdots, q_{\tilde{n}}] \in \reals^{n \times \tilde{n}}$.

Consider $\tilde{x} = \Lambda^{-\frac{1}{2}} Q^\top x \in \reals^{\tilde{n}}$. We have $\E [\tilde{x}\tilde{x}^\top] = \Lambda^{-\frac{1}{2}} Q^\top \Psi Q \Lambda^{-\frac{1}{2}} = I$. And we also have $x = Q \Lambda^{\frac{1}{2}} \tilde{x}$. 
For simplicity we let $\Theta = (Q \Lambda^{\frac{1}{2}})^\top$ and use $\theta_i$ to denote the $i$-th column vector of $\Theta$. Let $\Theta_1 = [\theta_1, \cdots, \theta_{a}] \in \reals^{\tilde{n} \times (n-a)}$ and $\Theta_2 = [\theta_{n-b+1}, \cdots, \theta_n] \in \reals^{\tilde{n} \times (n-b)}$, and let $\tilde{a} = \dm(\cl(\Theta_1))$ and $\tilde{b} = \dm(\cl(\Theta_2))$. Then we have $\dm(\cl(\Theta_1) \cap \cl(\Theta_2)) = \tilde{a} + \tilde{b} - \tilde{n}$. We can find $\tilde{n}$ vectors that form a basis of $\cl(\Theta)$, denoted by $\omega_{1}, \cdots, \omega_{\tilde{n}}$, such that  $\omega_1, \cdots, \omega_{\tilde{a}}$ and $\omega_{\tilde{n}-\tilde{b}+1}, \cdots, \omega_{\tilde{n}}$ form bases of $\cl(\Theta_1)$ and $\cl(\Theta_2)$ respectively, and $\omega_{\tilde{n}-\tilde{b}+1}, \cdots, \omega_{\tilde{a}}$ form a basis of $\cl(\Theta_1) \cap \cl(\Theta_2)$. 
Therefore, we let $\tilde{x}' = \Omega^\top \tilde{x}$ where $\Omega = [\omega_1, \cdots, \omega_{\tilde{n}}] \in \reals^{\tilde{n} \times \tilde{n}}$. 
% The transformations between $x$ and $\tilde{x}'$  are $x = Q \Lambda^{\frac{1}{2}} \Omega^{-\top} \tilde{x}'$ and $\tilde{x}' = \Omega^{\top} \Lambda^{-\frac{1}{2}} Q^\top x$ (Note that such transformation ). 
And from the construction process we have $\forall i \in \{1, \cdots, a\}$, $x_i$ can be expressed as a linear combination of $\tilde{x}'_1, \cdots, \tilde{x}'_{\tilde{a}}$; $\forall i \in \{n-b+1, \cdots, n\}$, $x_i$ can be expressed as a linear combination of $\tilde{x}'_{\tilde{n}-\tilde{b}+1}, \cdots, \tilde{x}'_{\tilde{n}}$. The same conclusion still holds if we switch $n, a, b, x$ and $\tilde{n}, \tilde{a}, \tilde{b}, \tilde{x}'$.

We define $\tilde{\Psi}$ = $\E[\tilde{x}' \tilde{x}'^\top]$. From the construction process, it is obvious that $\text{rank}(\tilde{\Psi})=n$. Moreover, $\forall i \in \{3, 4\}$, we define $\tilde{K}^{(i)} = K^{(i)} Q \Lambda^{\frac{1}{2}} \Omega^{-\top} \in \reals^{m_i \times \tilde{n}}$, and we have $\tilde{K}^{(i)} \tilde{x}' = K^{(i)} x$. In this way, we transformed $TaskAwareCoding(n, \Psi, a, b, Z, K^{(3)}, K^{(4)})$ to $TaskAwareCoding(\tilde{n}, \tilde{\Psi}, \tilde{a}, \tilde{b}, Z, \tilde{K}^{(3)}, \tilde{K}^{(4)})$.

Let $\{\tilde{E}^{(i,j)}|\forall (i,j) \in \mathcal{E}'\} \cup \{\tilde{D}^{(i)}|\forall i \in \{3,4\}\}$ be a solution for $TaskAwareCoding(\tilde{n}, \tilde{\Psi}, \tilde{a}, \tilde{b}, Z, \tilde{K}^{(3)}, \tilde{K}^{(4)})$. We can find $\{E^{(i,j)}|\forall (i,j) \in \mathcal{E}'\} \cup \{D^{(i)}|\forall i \in \{3,4\}\}$ for $TaskAwareCoding(n, \Psi, a, b, Z, K^{(3)}, K^{(4)})$, such that $\phi^{(i,j)} = \tilde{\phi}^{(i,j)}$, $\forall (i, j) \in \mathcal{E}$. For encoder paramters $\{E^{(i,j)}|\forall (i,j) \in \mathcal{E}'\}$, we take $E^{(1,3)}$ as an example. We let $E^{(1,3)} = \tilde{E}^{(1,3)} M$ where $M \in \reals^{\tilde{a} \times a}$ represents a linear transformation from $x_1, \cdots, x_a$ to $\tilde{x}'_1, \cdots, \tilde{x}'_{\tilde{a}}$. Moreover, for decoder parameters we let $D^{(i)} = Q \Lambda^{\frac{1}{2}} \Omega^{-\top} \tilde{D}^{(i)}$, $\forall i \in \{3, 4\}$. Therefore, $\forall i \in \{3,4\}$, we have $\tilde{K}^{(i)} \hat{\tilde{x}}' = K^{(i)} \hat{x}$. This implies the associated overall task losses for these two problems with these two sets of encoder and decoder parameters are the same.

Similarly, we can also transform a solution $\{E^{(i,j)}|\forall (i,j) \in \mathcal{E}'\} \cup \{D^{(i)}|\forall i \in \{3,4\}\}$ for $TaskAwareCoding(n, \Psi, a, b, Z, K^{(3)}, K^{(4)})$ to a set of parameters $\{\tilde{E}^{(i,j)}|\forall (i,j) \in \mathcal{E}'\} \cup \{\tilde{D}^{(i)}|\forall i \in \{3,4\}\}$ for $TaskAwareCoding(\tilde{n}, \tilde{\Psi}, \tilde{a}, \tilde{b}, Z, \tilde{K}^{(3)}, \tilde{K}^{(4)})$, and obtain the same conclusion.

This implies that the optimal overall task losses for these two problems are equal, and the above transformation from an optimal solution for one problem actually yields an optimal solution for the other.
\end{proof}

Therefore, in the rest of the analysis, we simply assume $\text{rank}(\Psi) = n$, and we let the Cholesky decomposition of $\Psi$ be $LL^\top$, where $L \in \reals^{n \times n}$. Moreover, notice that $\phi^{(i,j)}$ is a linear transformation from $x$ and hence is also a linear transformation from $L^{-1} x$. Therefore, for the convenience of the following analysis we let $\phi^{(i,j)} = \Phi^{(i,j)\top} L^{-1} x$ where $\Phi^{(i,j)} \in \reals^{n \times Z}$ is a transformation matrix. Furthermore, we assume $Z \leq n$, or else the network bandwidth is enough to make $\mathcal{L}^*_{\text{total}} = 0$.
% is not bottlenecked and the optimal solution is trivial. 

For task matrix $K^{(i)}$, $\forall i \in \{3, 4\}$, we define Gram matrix $S^{(i)} = L^\top K^{(i)\top} K^{(i)} L \in \reals^{n \times n}$. Moreover, let the eigen-values in descending order and the corresponding normalized eigen-vectors of $S^{(i)}$ be $\mu^{(i)}_{1}, \mu^{(i)}_{2}, \cdots, \mu^{(i)}_{n}$ and $u^{(i)}_{1}, u^{(i)}_{2}, \cdots, u^{(i)}_{n}$, respectively. Since node 3 receives $[\phi^{(1,3)\top}, \phi^{(5,6)\top}]^\top$ which has $2Z$ dimensions, according to PCA, 
% (as in 
% % Sec. \ref{subsec:preliminary_pca}
% Sec. \ref{sec:preliminaries}),
we have $\mathcal{L}^{(3)} \geq \sum_{j={2Z+1}}^n \mu^{(3)}_{j}$.
% (note if $2Z+1 > n$, the right-hand side is simply 0)
Similarly, $\mathcal{L}^{(4)} \geq \sum_{j={2Z+1}}^n \mu^{(4)}_{j}$. Therefore, $\mathcal{L}_{\text{total}} \geq \mathcal{L}_{\text{total,lb}}$ where $\mathcal{L}_{\text{total,lb}} = \sum_{i\in\{3,4\}} \sum_{j={2Z+1}}^n \mu^{(i)}_{j}$.

Ideally, we want to find an optimal solution associated with $\mathcal{L}_{\text{total,lb}}$, but $\mathcal{L}_{\text{total,lb}}$ may not be always achievable. Hence in the next two subsections, we focus on exploring the necessary conditions and sufficient conditions for $\mathcal{L}^*_{\text{total}} = \mathcal{L}_{\text{total,lb}}$.

For further analysis, $\forall i \in \{3, 4\}$, we define 
% $$U^{(i)} = [u^{(i)}_{1}, u^{(i)}_{2}, \cdots, u^{(i)}_{\min\{2Z,n\}}] \in \reals^{n \times \min\{2Z,n\}},$$
\begin{align}
U^{(i)} = [u^{(i)}_{1}, u^{(i)}_{2}, \cdots, u^{(i)}_{\min\{2Z,n\}}] \in \reals^{n \times \min\{2Z,n\}},
\end{align}
where the column vectors of $U^{(i)}$ are the top-$\min\{2Z,n\}$ normalized eigen-vectors of $S^{(i)}$. Making $\cl(U^{(3)}) \subseteq \cl([\Phi^{(1,3)}, \Phi^{(5,6)}])$ and $\cl(U^{(4)}) \subseteq \cl([\Phi^{(2,4)}, \Phi^{(5,6)}])$ is one way to achieve $\mathcal{L}_{\text{total,lb}}$. Moreover, we let $U^{(1)} \in \reals^{n \times a}$ and $U^{(2)} \in \reals^{n \times b}$ be matrices whose column vectors are the first $a$ and last $b$ column vectors of matrix $L$, respectively. The network topology constrains $\cl(\Phi^{(1,3)}), \cl(\Phi^{(1,5)}) \subseteq \cl(U^{(1)})$ and $\cl(\Phi^{(2,4)}), \cl(\Phi^{(2,5)}) \subseteq \cl(U^{(2)})$. Therefore, we say $\Phi^{(1,3)}$ is \textit{valid} if $\cl(\Phi^{(1,3)}) \subseteq \cl(U^{(1)})$, and $\Phi^{(2,4)}$ is \textit{valid} if $\cl(\Phi^{(2,4)}) \subseteq \cl(U^{(2)})$. 
On the other hand, any $\Phi^{(5,6)}$ is valid, since $\forall \Phi^{(5,6)} \in \reals^{n \times Z}$, $\exists \Phi^{(1,5)}, \Phi^{(2,5)}$ and $E^{(5,6)}$ s.t. $\Phi^{(5,6)\top} = E^{(5,6)} [\Phi^{(1,5)}, \Phi^{(2,5)}]^\top$, $\cl(\Phi^{(1,5)}) \subseteq \cl(U^{(1)})$, and $\cl(\Phi^{(2,5)}) \subseteq \cl(U^{(2)})$.

% $\cl([U^{(1)}, U^{(2)}]) = \reals^n$. 

Furthermore, we also let $r_+^{(i,j)} = \dm(\cl([U^{(i)}, U^{(j)}]))$ and $r_-^{(i,j)} = \dm(\cl(U^{(i)}) \cap \cl(U^{(j)}))$, $\forall i, j \in \{1, 2, 3, 4\}$, where $\dm(\cdot)$ is the dimension of a vector space.

\subsection{Necessary condition} \label{subsec:analysis_nc}

The following Theorem \ref{thm:nc} provides a necessary condition for achieving $\mathcal{L}_{\text{total,lb}}$ under a mild assumption. It constrains the dimensions of vector spaces from the perspective of network bandwidth.

\begin{theorem} \label{thm:nc}
Assume the eigen-gap $\mu^{(i)}_{\min\{2Z,n\}} - \mu^{(i)}_{\min\{2Z,n\}+1} > 0$ (define $\mu^{(i)}_{n+1}= 0$), $\forall i \in \{3, 4\}$. Then $\mathcal{L}_{\text{total,lb}}$ is achievable only when 
\begin{align}
r_+^{(3,4)} & \leq 3Z, \quad \textup{and}\label{eq:nc_1}\\
r_-^{(1,3)}, r_-^{(2,4)} & \geq \min\{Z,n-Z\}. \label{eq:nc_2}
\end{align}
\end{theorem}

\begin{proof}
If the eigen-gap $\mu^{(i)}_{\min\{2Z,n\}} - \mu^{(i)}_{\min\{2Z,n\}+1} > 0$, $\forall i \in \{3, 4\}$, then 
% according to 
% % the discussion in 
% % Sec. \ref{subsec:preliminary_pca}
% Sec. \ref{sec:preliminaries}, 
$\mathcal{L}_{\text{total,lb}}$ is achievable only when $\cl([\Phi^{(1,3)}, \Phi^{(5,6)}]) = \cl(U^{(3)})$ and $\cl([\Phi^{(2,4)}, \Phi^{(5,6)}]) = \cl(U^{(4)})$, which further implies $\cl([\Phi^{(1,3)}, \Phi^{(2,4)}, \Phi^{(5,6)}]) = \cl([U^{(3)}, U^{(4)}])$.

Notice that $\dm(\cl([\Phi^{(1,3)}, \Phi^{(2,4)}, \Phi^{(5,6)}]))\leq 3Z$. Thus when $r_+^{(3,4)} > 3Z$ it is impossible to make $\cl([\Phi^{(1,3)}, \Phi^{(2,4)}, \Phi^{(5,6)}]) = \cl([U^{(3)}, U^{(4)}])$. Hence $\mathcal{L}_{\text{total,lb}}$ is not achievable.

Moreover, if $r_-^{(1,3)} < \min\{Z,n-Z\}$, then $\dm(\cl([\Phi^{(1,3)}, \Phi^{(5,6)}])) \leq r_-^{(1,3)} + Z < \min\{2Z,n\} = \dm(\cl(U^{(3)}))$, which means we cannot make $\cl([\Phi^{(1,3)}, \Phi^{(5,6)}]) = \cl(U^{(3)})$. So $\mathcal{L}_{\text{total,lb}}$ is not achievable. Similarly, $\mathcal{L}_{\text{total,lb}}$ is not achievable when $r_-^{(2,4)} < \min\{Z,n-Z\}$.
% Thus the theorem is proved. 
\end{proof}

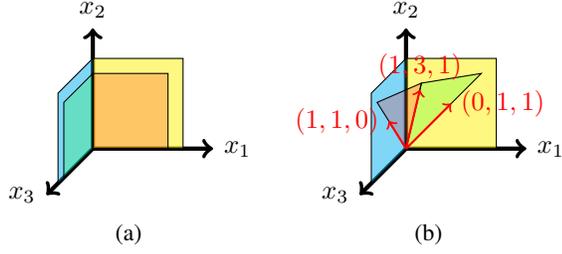
\begin{figure}[t]
\centering

\begin{tikzpicture}[state/.style = {shape=rectangle,thick,draw,minimum width=+1.7ex,minimum height=+1.7ex,inner sep=+.1pt}]
\matrix[column sep=1ex]{
\path (0,0.5em) node[state,fill=yellow, fill opacity=0.5,label=right:$\cl(U^{(1)})$]{};&
\path (0,0.5em) node[state,fill=cyan, fill opacity=0.5, label=right:$\cl(U^{(2)})$]{};&
\path (0,0.5em) node[state,fill=green, fill opacity=0.2, label=right:$\cl(U^{(3)})$]{};&
\path (0,0.5em) node[state,fill=red,fill opacity=0.2, label=right:$\cl(U^{(4)})$]{};\\
};
\end{tikzpicture} 
\subfigure[]{
\begin{tikzpicture}[scale=0.4]
\pgfmathsetmacro{\lenmax}{4}
\pgfmathsetmacro{\lenbg}{3}
\pgfmathsetmacro{\len}{2.5}
\draw[->,ultra thick] (0,0,0)--(\lenmax,0,0) node[right]{$x_1$};
\draw[->,ultra thick] (0,0,0)--(0,\lenmax,0) node[above]{$x_2$};
\draw[->,ultra thick] (0,0,0)--(0,0,\lenmax) node[left]{$x_3$};
\draw[black,fill=yellow, fill opacity=0.5] (0,0,0) -- (\lenbg,0,0) -- (\lenbg,\lenbg,0) -- (0,\lenbg,0) -- cycle;
\draw[black,fill=cyan, fill opacity=0.5] (0,0,0) -- (0,\lenbg,0) -- (0,\lenbg,\lenbg) -- (0,0,\lenbg) -- cycle;
\draw[black,fill=red,fill opacity=0.2] (0,0,0) -- (\len,0,0) -- (\len,\len,0) -- (0,\len,0) -- cycle;
\draw[black,fill=green,fill opacity=0.2] (0,0,0) -- (0,\len,0) -- (0,\len,\len) -- (0,0,\len) -- cycle;
\end{tikzpicture}
\label{fig:exp1}
}
\subfigure[]{
\begin{tikzpicture}[scale=0.4]
\pgfmathsetmacro{\lenmax}{4}
\pgfmathsetmacro{\lenbg}{3}
\pgfmathsetmacro{\len}{2.5}
\draw[->,ultra thick] (0,0,0)--(\lenmax,0,0) node[right]{$x_1$};
\draw[->,ultra thick] (0,0,0)--(0,\lenmax,0) node[above]{$x_2$};
\draw[->,ultra thick] (0,0,0)--(0,0,\lenmax) node[left]{$x_3$};
\draw[black,fill=yellow, fill opacity=0.5] (0,0,0) -- (\lenbg,0,0) -- (\lenbg,\lenbg,0) -- (0,\lenbg,0) -- cycle;
\draw[black,fill=cyan, fill opacity=0.5] (0,0,0) -- (0,\lenbg,0) -- (0,\lenbg,\lenbg) -- (0,0,\lenbg) -- cycle;
\draw[black,fill=red,fill opacity=0.2] (0,0,0) -- (\len/3,\len,\len/3) -- (0,\len,\len) -- cycle;
\draw[black,fill=green,fill opacity=0.2] (0,0,0) -- (\len,\len,0) -- (\len/3,\len,\len/3) -- cycle;
\draw[->,thick,red] (0,0,0)--(0.9*\len/3,0.9*\len,0.9*\len/3) node[above]{$(1,3,1)$};
\draw[->,thick,red] (0,0,0)--(0.6*\len,0.6*\len,0) node[right]{ $(0,1,1)$};
\draw[->,thick,red] (0,0,0)--(0,0.6*\len,0.6*\len) node[left]{$(1,1,0)$};
\end{tikzpicture}
\label{fig:exp2}
}

\caption{Two illustrative examples for Theorem \ref{thm:nc}, where the left one doesn't achieve $\mathcal{L}_{\text{total,lb}}$ while the right one does. }
\label{fig:exp}
\end{figure}

To show the conditions in Theorem \ref{thm:nc} are only necessary but not sufficient, we present two examples in Fig. \ref{fig:exp}, where the left one doesn't achieve $\mathcal{L}_{\text{total,lb}}$ while the right one does. Here we have $n=3$, $\Psi = I$, $Z = 1$, $a = b = 2$. 
And we also assume eigen-gap $\mu^{(i)}_{2} - \mu^{(i)}_{3} > 0$, $\forall i \in \{3, 4\}$. Therefore, to achieve $\mathcal{L}_{\text{total,lb}}$, we must have $\cl([\Phi^{(1,3)}, \Phi^{(5,6)}]) = \cl(U^{(3)})$ and $\cl([\Phi^{(2,4)}, \Phi^{(5,6)}]) = \cl(U^{(4)})$.
In Fig. \ref{fig:exp1}, we assume $u^{(3)}_{1} = u^{(4)}_{1} = [0, 1, 0]^\top$, $u^{(3)}_{2} = [0, 0, 1]^\top$ and $u^{(4)}_{2} = [1, 0, 0]^\top$. So we have $r_+^{(3,4)} = 3$ and $r_-^{(1,3)} = r_-^{(1,4)} = 1$. The conditions in Theorem \ref{thm:nc} are satisfied, but we cannot make $\cl([\Phi^{(1,3)}, \Phi^{(5,6)}]) = \cl(U^{(3)})$ and $\cl([\Phi^{(2,4)}, \Phi^{(5,6)}]) = \cl(U^{(4)})$ simultaneously, and hence $\mathcal{L}_{\text{total,lb}}$ is not achievable. 
% This is because, suppose $\cl([\Phi^{(1,3)}, \Phi^{(6,3)}]) = \cl(U^{(3)})$ and $\cl([\Phi^{(2,4)}, \Phi^{(6,4)}]) = \cl(U^{(4)})$, then $\Phi^{(6,3)} \in \cl(U^{(3)})$ must be on the plane $x_1 = 0$, and $\Phi^{(6,4)} \in \cl(U^{(4)})$ must be on the plane $x_3 = 0$. Thus $\Phi^{(6,3)} = \Phi^{(6,4)}$ must have the form of $ [0, \beta, 0]^\top$ where $\beta$ is a non-zero real number. To make $\cl([\Phi^{(2,4)}, \Phi^{(6,4)}]) = \cl(U^{(4)})$ we must have $\Phi^{(2,4)} = [\alpha, \beta', 0]^\top$ where $\alpha$ is a non-zero real number, but $\Phi^{(2,4)} \nsubseteq \cl(U^{(2)})$ and hence is invalid.
In Fig. \ref{fig:exp2}, we assume $u^{(3)}_{1} = u^{(4)}_{1} = \frac{1}{\sqrt{11}}[1, 1, 3]^\top$, $u^{(3)}_{2} = \frac{1}{\sqrt{66}} [4, -7, 1]^\top$ and $u^{(4)}_{2} = \frac{1}{\sqrt{66}} [-7, 4, 1]^\top$. For $\Phi^{(1,3)} = [0, 1, 1]^\top$, $\Phi^{(2,4)} = [1, 0, 1]^\top$ and $\Phi^{(5,6)} = [1, 3, 1]^\top$ (which are vectors belong to the intersections of different column spans), $\mathcal{L}_{\text{total,lb}}$ is achievable because $\cl([\Phi^{(1,3)}, \Phi^{(5,6)}]) = \cl(U^{(3)})$ and $\cl([\Phi^{(2,4)}, \Phi^{(5,6)}]) = \cl(U^{(4)})$.

% On thing interesting to notice is even though $\Phi^{(1,3)}$ is neither $u^{(3)}_{1}$ and $u^{(3)}_{2}$ (and similarly for $\Phi^{(2,4)}$), we are still able to achieve $\mathcal{L}_{\text{total,lb}}$.

%%%%%%%%%%%%
%TABLE
%%%%%%%%%%%%

% \begin{table*}[t]
% \caption{Sufficient conditions for achieving $\mathcal{L}_{\text{total,lb}}$.}
% \label{tab:sf}
% \vskip 0.15in
% \begin{center}
% \begin{tabular}{C{0.11\linewidth} | C{0.20\linewidth} | C{0.6\linewidth}}
% \toprule
% \multirow{2}{*}{Observations} & \multicolumn{2}{c}{Tasks}\\
% \cline{2-3}
%  & Same ($S^{(3)} = S^{(4)}$) & Not Same\\
% \midrule
% Same ($a=b=n$) &  \makecell{Corollary \ref{cor:sf_basic}: \\ Always achievable} & Corollary \ref{cor:sf_same_ob}: $r_+^{(3,4)} \leq 3Z$ \\
% \hline
% Not Same & \makecell{Corollary \ref{cor:sf_same_task_1}: \\$\cl(U^{(3)}) \subseteq$\\$\cl(U^{(1)}) \cap \cl(U^{(2)})$; \\\textbf{or} Corollary \ref{cor:sf_same_task_2}: \\ $n \leq Z + \min\{a, b\}$} & 
% \makecell{
% The following 1 and 2a (Theorem \ref{thm:sf2}) \textbf{or} 1 and 2b (Theorem \ref{thm:sf3}):\\
% 1) $r_+^{(3,4)} \leq 3Z$, \\
% $\cl(U^{(3)}) = \sn((\cl(U^{(1)}) \cap \cl(U^{(3)})) \cup (\cl(U^{(3)}) \cap \cl(U^{(4)})))$,\\
% $\cl(U^{(4)}) = \sn((\cl(U^{(2)}) \cap \cl(U^{(4)})) \cup (\cl(U^{(3)}) \cap \cl(U^{(4)})))$;\\ 
% 2a) $\cl(U^{(3)}) \cap \cl(U^{(4)}) \subseteq \cl(U^{(1)}) \cap \cl(U^{(2)})$;\\ 
% 2b) $n \leq Z + \min\{a, b\}$.
% }
% \\
% \bottomrule
% \end{tabular}
% \end{center}
% \vskip -0.1in
% \end{table*}

\subsection{Sufficient Conditions} \label{subsec:analysis_sc}

We have seen that constrain the dimensions of vector spaces, as in Theorem \ref{thm:nc}, is not enough to achieve $\mathcal{L}_{\text{total,lb}}$. 
In the following theorem, we add a requirement of the data dependencies between different $U^{(i)}$'s on top of the necessary conditions, and hence the achievability of $\mathcal{L}_{\text{total,lb}}$ is guaranteed.

% New sufficient condition
\begin{theorem}\label{thm:sf0}
If Eq. (\ref{eq:nc_1}) and (\ref{eq:nc_2}) hold, and 
% \begin{align}
% \cl(U^{(3)}) = \sn(
% & (\cl(U^{(1)}) \cap \cl(U^{(3)})) \cup (\cl(U^{(3)}) \cap \cl(U^{(4)}))), \label{eq:sf_1}\\
% \cl(U^{(4)}) = \sn(
% & (\cl(U^{(2)}) \cap \cl(U^{(4)})) \cup (\cl(U^{(3)}) \cap \cl(U^{(4)}))), \label{eq:sf_2}
% \end{align}
\begin{align}
\cl(U^{(3)}) = \sn(
& (\cl(U^{(1)}) \cap \cl(U^{(3)})) \cup \notag\\
& (\cl(U^{(3)}) \cap \cl(U^{(4)}))), \label{eq:sf_1}\\
\cl(U^{(4)}) = \sn(
& (\cl(U^{(2)}) \cap \cl(U^{(4)})) \cup \notag\\
& (\cl(U^{(3)}) \cap \cl(U^{(4)}))), \label{eq:sf_2}
\end{align}
then $\mathcal{L}_{\text{total,lb}}$ is achievable.
\end{theorem}

Eq. (\ref{eq:sf_1}) (and similarly for Eq. (\ref{eq:sf_2})) has the following interpretation: we can find vectors in $\cl(U^{(1)})$ that extend a basis of $\cl(U^{(3)}) \cap \cl(U^{(4)})$ to a basis of $\cl(U^{(3)})$. This makes it possible for us to assign column vectors of $\Phi^{(1,3)}$ to achieve $\mathcal{L}_{\text{total,lb}}$ (which is not possible for Fig. \ref{fig:exp1}). 

% For the sake of clarity, we only provide the proof of Theorem \ref{thm:sf0} when $2Z \leq n$. The proof idea when $2Z > n$ is quite similar and we put it in Sec. \ref{subsec:proof_thm_sf0} due to space limit.

For the sake of clarity, we only provide the proof of Theorem \ref{thm:sf0} when $2Z \leq n$. The proof idea when $2Z > n$ is quite similar and we put it in the appendix due to space limit.

\begin{proof}
Since $2Z \leq n$, we have $r_-^{(3,4)} = 4Z - r_+^{(3,4)} \geq Z$, and $r_-^{(1,3)}, r_-^{(2,4)} \geq Z$ according to Eq. (\ref{eq:nc_1}) and (\ref{eq:nc_2}). We will construct valid $\Phi^{(1,3)}$, $\Phi^{(2,4)}$ and $\Phi^{(5,6)}$ such that $\mathcal{L}_{\text{total,lb}}$ is achievable. There are four steps for construction:

$\bullet$ i) According to Eq. (\ref{eq:sf_1}), we can find $2Z - r_-^{(3,4)}$ vectors in $\cl(U^{(1)}) \cap \cl(U^{(3)})$ that extend a basis of $\cl(U^{(3)}) \cap \cl(U^{(4)})$ to a basis of $\cl(U^{(3)})$. We let them be the first $2Z - r_-^{(3,4)}$ column vectors of $\Phi^{(1,3)}$. We similarly determine the first $2Z - r_-^{(3,4)}$ column vectors of $\Phi^{(2,4)}$.

$\bullet$ ii) Suppose space $\cl(U^{(1)}) \cap \cl(U^{(2)}) \cap \cl(U^{(3)}) \cap \cl(U^{(4)})$ has $p$ dimensions. We randomly choose $\min\{p, r_-^{(3,4)} - Z\}$ linear independent vectors from this space and assign them as some of the non-determined column vectors of $\Phi^{(1,3)}$ and $\Phi^{(2,4)}$. If $p \geq r_-^{(3,4)} - Z$, then all the column vectors of $\Phi^{(1,3)}$ and $\Phi^{(2,4)}$ are determined, and we will skip the following step iii.

$\bullet$ iii) The space $\cl(U^{(1)}) \cap \cl(U^{(3)}) \cap \cl(U^{(4)})$ has at least $r_-^{(3,4)} + r_-^{(1,3)} -\dm(\cl(U^{(3)})) \geq  r_-^{(3,4)} - Z$ dimensions. Similarly, the space $\cl(U^{(2)}) \cap \cl(U^{(3)}) \cap \cl(U^{(4)})$ has at least $r_-^{(3,4)} - Z$ dimensions. Hence we can choose $r_-^{(3,4)} - Z - p$ linear independent vectors from these two spaces respectively, such that they are independent to the $p$ vectors chosen in step ii. Moreover, these two sets of vectors are also independent since they do not belong to $\cl(U^{(1)}) \cap \cl(U^{(2)}) \cap \cl(U^{(3)}) \cap \cl(U^{(4)})$. Hence we let them be the remaining non-determined column vectors of $\Phi^{(1,3)}$ and $\Phi^{(2,4)}$ respectively, and we let the first $r_-^{(3,4)} - Z - p$ non-determined column vectors of $\Phi^{(5,6)}$ be the pair-wise sums\footnote{
Here pair-wise sums of vectors $\xi_1, \xi_2, \cdots, \xi_p$ and vectors $\chi_1, \chi_2, \cdots, \chi_p$ are vectors $\xi_1 + \chi_1, \xi_2 + \chi_2, \cdots, \xi_p + \chi_p$. With vectors $\xi_1, \xi_2, \cdots, \xi_p$ and vectors $\xi_1 + \chi_1, \xi_2 + \chi_2, \cdots, \xi_p + \chi_p$, one can decode vectors $\chi_1, \chi_2, \cdots, \chi_p$ through $\chi_i = -\xi_i + (\xi_i + \chi_i)$, $\forall i \in \{1, 2, \cdots, p\}$. 
}
of these two sets of vectors.

$\bullet$ iv) We make the remaining non-determined column vectors of $\Phi^{(5,6)}$ be the vectors that extend the vectors chosen in step ii and iii to a basis of $\cl(U^{(3)}) \cap \cl(U^{(4)})$.

The constructed $\Phi^{(1,3)}$ and $\Phi^{(2,4)}$ are valid, and we also have $\cl([\Phi^{(1,3)}, \Phi^{(5,6)}]) = \cl(U^{(3)})$ and $\cl([\Phi^{(2,4)}, \Phi^{(5,6)}]) = \cl(U^{(4)})$. Therefore, $\mathcal{L}_{\text{total,lb}}$ is achievable.
\end{proof}

We further have the following two corollaries.

\begin{corollary}\label{thm:sf2}
If Eq. (\ref{eq:nc_1}), (\ref{eq:sf_1}) and (\ref{eq:sf_2}) hold, and
\begin{align}
\cl(U^{(3)}) \cap \cl(U^{(4)}) \subseteq \cl(U^{(1)}) \cap \cl(U^{(2)}),\label{eq:sf_3}
\end{align} 
then $\mathcal{L}_{\text{total,lb}}$ is achievable.
\end{corollary}

\begin{proof}
According to Eq. (\ref{eq:sf_1}) and (\ref{eq:sf_3}), $ \cl(U^{(3)}) \subseteq \cl(U^{(1)})$, which means $r_-^{(1,3)} = \min\{2Z, n\}$. Similarly, $r_-^{(2,4)} = \min\{2Z, n\}$. Hence Eq. (\ref{eq:nc_2}) is satisfied, and $\mathcal{L}_{\text{total,lb}}$ is achievable. 
% according to Theorem \ref{thm:sf0}.
\end{proof}

In fact, in Corollary \ref{thm:sf2}, we don't even need network coding to achieve $\mathcal{L}_{\text{total,lb}}$. If $2Z \leq n$, then step iii, in which network coding is needed, will always be skipped in the proof of Theorem \ref{thm:sf0}; if $2Z > n$, Eq. (\ref{eq:sf_3}) implies $\cl(U^{(1)}) = \cl(U^{(2)}) = \reals^n$, and network coding is also not needed.

\begin{corollary}\label{thm:sf3}
If Eq. (\ref{eq:nc_1}), (\ref{eq:sf_1}) and (\ref{eq:sf_2}) hold, and
\begin{align}
n \leq Z + \min\{a, b\}, \label{eq:sf_4}
\end{align}
then $\mathcal{L}_{\text{total,lb}}$ is achievable.
\end{corollary}

\begin{proof}
According to Eq. (\ref{eq:sf_4}), we know $\dim(\cl(U^{(1)})) = a \geq n - Z$. Since $\dim(\cl(U^{(3)})) = \min\{2Z, n\}$, we know $r_-^{(1,3)} \geq \dim(\cl(U^{(1)})) + \dim(\cl(U^{(3)})) - n \geq \min\{Z, n - Z\}$. Similarly, $r_-^{(1,3)} \geq \min\{Z, n - Z\}$. Hence Eq. (\ref{eq:nc_2}) is satisfied, and $\mathcal{L}_{\text{total,lb}}$ is achievable. 
% according to Theorem \ref{thm:sf0}.
\end{proof}

\section{Algorithm}\label{sec:algorithm}

% In Sec. \ref{subsec:analysis_sc} 
In the last section
we have discussed the sufficient conditions for achieving $\mathcal{L}_\text{total,lb}$, and corresponding optimal encoder and decoder parameters can be determined analytically. In the general case when these sufficient conditions are not satisfied, we resort to standard gradient descent algorithm to determine the encoder and decoder parameters jointly. The encoders and decoders are connected as per network information flow (i.e., Eq. (\ref{eq:encoding_1})-(\ref{eq:decoding_2})).
We initialize encoder and decoder parameters randomly and update them for multiple epochs. In each epoch, we update $E^{(i, j)}$ and $D^{(i)}$ through back-propagation as follows:
\begin{align}
E^{(i, j)} \leftarrow E^{(i, j)} - \eta \frac{\nabla \mathcal{L}_\text{total}}{\nabla E^{(i, j)}}, \quad \forall (i, j) \in \mathcal{E}'; \\
D^{(i)} \leftarrow D^{(i)} - \eta \frac{\nabla \mathcal{L}_\text{total}}{\nabla D^{(i)}}, \quad \forall i \in \{3, 4\}
\end{align}
where $\eta$ is the learning rate.

\begin{figure}[t]
\centering

\begin{tikzpicture}[scale=0.5]
    \begin{customlegend}[
    legend entries={Task-aware coding,Task-aware no coding,Task-agnostic coding (appendix),Coding benchmark,Lower bound $\mathcal{L}_{\text{total,lb}}$},
    legend columns=2,
    legend style={
            /tikz/column 1/.style={
                column sep=5pt,
            },
            /tikz/column 2/.style={
                column sep=5pt,
                font=\scriptsize,
            },
            /tikz/column 3/.style={
                column sep=5pt,
            },
            /tikz/column 4/.style={
                column sep=5pt,
                font=\scriptsize,
            },
            /tikz/column 5/.style={
                column sep=5pt,
            },
            /tikz/column 6/.style={
                column sep=5pt,
                font=\scriptsize,
            },
            /tikz/column 7/.style={
                column sep=5pt,
            },
            /tikz/column 8/.style={
                column sep=5pt,
                font=\scriptsize,
            },
            /tikz/column 9/.style={
                column sep=5pt,
            },
            /tikz/column 10/.style={
                font=\scriptsize,
            },
        }]
    \addlegendimage{blue,mark=*,line width=1pt}
    \addlegendimage{red,mark=x,line width=1pt}
    \addlegendimage{green!60!black,mark=triangle,line width=1pt}
    \addlegendimage{brown,mark=diamond,line width=1pt}
    \addlegendimage{black,dashed, mark=none,line width=1pt}
    \end{customlegend}
\end{tikzpicture}

% \subfigure[]{
% \input{figs_tikz/synthetic_test_loss_r34+.tex}
% \label{fig:synthetic_loss_r34+}
% }
% \subfigure[]{
% \input{figs_tikz/synthetic_test_loss_a.tex}
% \label{fig:synthetic_loss_a}
% }

\begin{minipage}[t]{0.49\linewidth}
\strut\vspace*{-\baselineskip}\newline
\subfigure[]{
\begin{tikzpicture}[scale=0.55]
    \pgfplotsset{normalsize,samples=10}
    \begin{axis}[ height=6cm, width=8cm,
                  % legend style={at={(0.5,0.25)},
                  %              anchor=north,legend columns=3},
                  % legend pos=outer north east,
                  % legend entries = {task aware coding,
                  %                   task aware no coding,
                  %                   task agnostic coding,
                  %                   benchmark,
                  %                   lower bound},
                  % xmin=0,xmax=1.0,ymin=0,ymax=0.035,
                  % xtick distance=1,
                  xlabel={$r^{(3,4)}_+$},
                  ylabel={Overall task loss $\mathcal{L}_{\text{total}}$},
                  xlabel style={font=\Large},
                  ylabel style={font=\large},
                  scaled y ticks=false,
                  every axis plot/.append style={ultra thick},
                  mark size=3pt ]
            \addplot [blue,mark=*, error bars/.cd, y dir=both, y explicit,
                      error bar style={line width=2pt,solid},
                      error mark options={line width=1pt,mark size=4pt,rotate=90}]
                      table [x=x, y=y, y error=y-err]{%
                        x y y-err
16 6.911812548195489e-07 3.073093965247263e-07
17 5.066640246184788e-06 4.731035325567704e-06
18 1.4494722576517923e-08 8.568561623689385e-09
19 3.7463711480241007e-07 5.220257716285899e-07
20 2.449196049536813e-07 3.004506629594412e-07
21 1.040607480574976e-08 5.747604562830079e-09
22 4.4599350823704945e-08 2.4596522828269098e-08
23 2.4981650065514156e-07 2.1498828483884266e-07
24 0.003985947765315031 0.0018782044144822862
25 0.6287877687241475 0.3296177718246637
26 1.3985435289708184 0.5970179153308545
27 2.0173303157938416 0.5896322900229524
28 2.7016288557100814 0.7252355307733556
29 3.2574517610267923 0.666716962805083
30 3.950915976422186 0.8473061459132485
31 4.470203985259172 0.7335067487170107
32 5.222534629085299 0.6334240631758826
                      };
            \addplot [red,mark=x, error bars/.cd, y dir=both, y explicit,
                      error bar style={line width=2pt,solid},
                      error mark options={line width=1pt,mark size=4pt,rotate=90}]
                      table [x=x, y=y, y error=y-err]{%
                        x y y-err
16 1.742752311752271e-29 5.730866711219217e-30
17 2.2890390381711667e-06 3.208716095576786e-06
18 1.2096275430612091e-13 1.0219516981021993e-13
19 0.005503899890213851 0.005503531892064301
20 1.938916179175619e-09 2.4857361113876252e-09
21 1.211837155947631e-06 1.0760555426549954e-06
22 4.933278383830637e-06 3.3127497990603083e-06
23 4.90165416283887e-10 3.429956141126788e-10
24 6.546608913181524e-05 4.8794794796865374e-05
25 0.6649857977148086 0.35021088295879266
26 1.439081132176372 0.6905948150178822
27 2.088807174345342 0.6198602793712682
28 2.7754194854035577 0.7106754686004189
29 3.4392795633906537 0.7687401368030804
30 4.062084090001905 0.8930110456623133
31 4.679884050573273 0.7165165233855438
32 5.2653395974012 0.8237096580399346
                      };
%             \addplot [green!60!black,mark=triangle, error bars/.cd, y dir=both, y explicit,
%                       error bar style={line width=2pt,solid},
%                       error mark options={line width=1pt,mark size=4pt,rotate=90}]
%                       table [x=x, y=y, y error=y-err]{%
%                         x y y-err
% 16 172.8375987897491 47.724847121914344
% 17 94.26239487799364 19.671373634057638
% 18 174.18427743692456 28.564440985114015
% 19 153.83579627635987 44.25555735925807
% 20 130.1169557970644 22.319986174234273
% 21 158.7624182280638 34.454522365565204
% 22 149.8457812863585 32.44706573459772
% 23 156.46833736789608 31.478330604830283
% 24 148.00023160061866 28.802687306532757
% 25 121.31254481189269 19.927416701892422
% 26 154.0859798561758 26.848979893047566
% 27 130.38710495072678 23.84811046094067
% 28 168.49931352001445 31.56676558030003
% 29 163.96303725166513 42.40842934528732
% 30 148.68978175424812 30.336870944823954
% 31 156.41032519017318 27.607111116237732
% 32 103.32288648321963 17.970249843232764
%                       };
            \addplot [brown,mark=diamond, error bars/.cd, y dir=both, y explicit,
                      error bar style={line width=2pt,solid},
                      error mark options={line width=1pt,mark size=4pt,rotate=90}]
                      table [x=x, y=y, y error=y-err]{%
                        x y y-err
16 1.5833892589714623e-29 4.635687923215124e-30
17 1.1260613734286689 0.38288194039503937
18 1.3631532064732141 0.44691973668179164
19 1.6841791041242855 0.48955672260661887
20 2.4659839428046615 0.7137615365849697
21 2.9211683011400527 0.7623771591516395
22 3.202657351318381 0.7498832656669853
23 3.486696320035632 0.7257083609174176
24 3.7184039463777374 0.6532013575787966
25 4.407380780024825 0.8474760654931122
26 5.413272170163493 0.9942789025141056
27 4.792027967436242 0.8069740159750458
28 5.069207318588853 0.7418878368767738
29 5.385768487248315 0.824117134123108
30 5.7433438333280336 0.8603413886063387
31 5.4370998900855 0.7808543963950577
32 5.265337455493307 0.8237117375582347
                      };
            \addplot [black,dashed, mark=none, error bars/.cd, y dir=both, y explicit,
                      error bar style={line width=2pt,solid},
                      error mark options={line width=1pt,mark size=4pt,rotate=90}]
                      table [x=x, y=y]{%
                        x y
16 0
32 0
                      };
    \end{axis}
\end{tikzpicture}
\label{fig:synthetic_loss_r34+}
}
\end{minipage}
\begin{minipage}[t]{0.49\linewidth}
\strut\vspace*{-\baselineskip}\newline
\subfigure[]{
\begin{tikzpicture}[scale=0.55]
    \pgfplotsset{normalsize,samples=10}
    \begin{axis}[ height=6cm, width=8cm,
                  % legend style={at={(0.5,0.25)},
                  %              anchor=north,legend columns=3},
                  % legend pos=outer north east,
                  % legend entries = {task aware coding,
                  %                   task aware no coding,
                  %                   task agnostic coding,
                  %                   benchmark,
                  %                   lower bound},
                  % xmin=0,xmax=1.0,ymin=0,ymax=0.035,
                  % xtick distance=1,
                  xlabel={$a (=b)$},
                  xlabel style={font=\Large},
                  ylabel style={font=\large},
                  scaled y ticks=false,
                  every axis plot/.append style={ultra thick},
                  mark size=3pt ]
            \addplot [blue,mark=*, error bars/.cd, y dir=both, y explicit,
                      error bar style={line width=2pt,solid},
                      error mark options={line width=1pt,mark size=4pt,rotate=90}]
                      table [x=x, y=y, y error=y-err]{%
                        x y y-err
16 6.7887583596120065 1.3249825957190458
17 6.081541798097357 1.329619261431759
18 5.471139056021478 1.389162055664233
19 3.7577250566057714 1.1256888558544584
20 2.884418669358931 1.0126211404453778
21 1.8355849644996638 0.8183213708915495
22 0.8821168576103318 0.3862255576440212
23 0.4083817648606587 0.2461480301702073
24 1.067947262089434e-11 2.5716226398142827e-12
25 9.46066349056175e-12 1.6935688272174662e-12
26 1.1577685293234012e-11 2.4185845724366454e-12
27 7.709805233056714e-08 1.0003009778788454e-07
28 1.417006439065208e-11 3.3575133693265813e-12
29 1.3213822034014166e-11 7.773963355986723e-12
30 7.817713672374309e-12 1.7577774127902056e-12
31 9.137685871359337e-12 2.6289755484753194e-12
32 9.081275297734164e-12 1.733473704279905e-12
                      };
            \addplot [red,mark=x, error bars/.cd, y dir=both, y explicit,
                      error bar style={line width=2pt,solid},
                      error mark options={line width=1pt,mark size=4pt,rotate=90}]
                      table [x=x, y=y, y error=y-err]{%
                        x y y-err
16 37.95434945235194 9.226613349288886
17 34.44220156560836 8.794525333239863
18 31.11666840229244 8.347713724698014
19 27.361329927213163 8.255567990106641
20 24.732487391550794 8.032952735060224
21 20.33029476627106 7.352346514783946
22 16.80647585981944 7.028196663272082
23 14.701982840256784 6.918209261016229
24 11.03403817465439 5.921611040509638
25 9.513402956501773 5.690230502401569
26 7.5491721214631085 5.165000834939477
27 6.111124602561819 4.871116008134628
28 3.910744764790789 3.9121252516009855
29 1.5215836375561946 2.132926536415867
30 3.07216461573355e-07 4.306223952321878e-07
31 1.4677685774388845e-08 2.0412965465358896e-08
32 0.0005916304464809701 0.00018367132264438609
                      };
%             \addplot [green!60!black,mark=triangle, error bars/.cd, y dir=both, y explicit,
%                       error bar style={line width=2pt,solid},
%                       error mark options={line width=1pt,mark size=4pt,rotate=90}]
%                       table [x=x, y=y, y error=y-err]{%
%                         x y y-err
% 16 180.94475147145315 32.67183961020032
% 17 253.3906242465415 88.10778236802467
% 18 326.2451391348999 206.64532719414515
% 19 192.4431457469351 65.41601142742745
% 20 259.0869113563947 66.85752102100702
% 21 221.89017561068948 74.34429999873869
% 22 380.6332677204433 225.33149433141952
% 23 293.023320169583 97.74720100285599
% 24 303.7692319400638 107.56738477163343
% 25 235.07202106690488 63.49304675012292
% 26 307.9263476944641 101.79277593432704
% 27 295.70907658554927 118.40815984547349
% 28 267.469221296904 76.5472062421604
% 29 319.0429011516379 165.94029539651308
% 30 322.08980821016564 121.05731711122395
% 31 324.983549714093 132.45662438199093
% 32 337.51753378774777 172.94070340715322
%                       };
            \addplot [brown,mark=diamond, error bars/.cd, y dir=both, y explicit,
                      error bar style={line width=2pt,solid},
                      error mark options={line width=1pt,mark size=4pt,rotate=90}]
                      table [x=x, y=y, y error=y-err]{%
                        x y y-err
16 18.486370519701005 2.554308048686343
17 17.566420932595655 2.646875405570866
18 16.786696249531595 2.73724470998191
19 15.373021622232717 3.1926975733545393
20 13.646280515199452 3.1634653399850095
21 10.943880382314727 2.7131735754466595
22 9.674201116063395 2.523222019557381
23 7.943644220784298 2.053148526004178
24 5.540677786405476 1.3818451471976547
25 4.829952306703767 1.242692871758804
26 4.070351294157554 1.6885254549526247
27 4.070351294157554 1.6885254549526243
28 4.070351294157554 1.6885254549526238
29 4.070351294157554 1.6885254549526243
30 4.070351294157554 1.6885254549526214
31 4.070351294157554 1.688525454952624
32 4.070351294157554 1.688525454952625
                      };
            \addplot [black,dashed, mark=none, error bars/.cd, y dir=both, y explicit,
                      error bar style={line width=2pt,solid},
                      error mark options={line width=1pt,mark size=4pt,rotate=90}]
                      table [x=x, y=y]{%
                        x y
16 0
32 0
                      };
    \end{axis}
\end{tikzpicture}
\label{fig:synthetic_loss_a}
}
\end{minipage}

\caption{Simulation result with synthetic data: overall task loss $\mathcal{L}_\text{total}$ under different $r_+^{(3,4)}$ (left) and different $a$ (right). The task losses for task-agnostic coding are too large and have to be put in a separate figure in the appendix.}
\label{fig:synthetic_evaluation}
\end{figure}
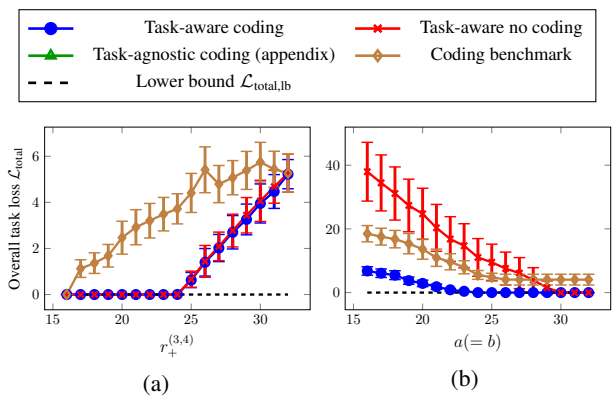

To show our algorithm converges to near-optimal solution for low-dimensional data and to verify our conclusions 
in the last section
% in Sec. \ref{subsec:analysis_sc} 
numerically, we run simulation with synthetic data for our task-aware coding approach and compare against three benchmark approaches. The benchmark approaches are: 1) \textbf{Task-aware no coding} approach, where network coding at node 5 is not allowed, i.e., each dimension of $\phi^{(5,6)}$ can only be a dimension of $\phi^{(1,5)}$ or $\phi^{(2,5)}$; 2) \textbf{Task-agnostic coding} approach (used in \citet{liu2020neural}), where the objective is to minimize the reconstruction loss at node 3 and 4, i.e., $K^{(3)} = K^{(4)} = I$; 3) \textbf{Task-aware coding benchmark} (abbreviated as \textbf{coding benchmark}) approach, which is also a task-aware coding approach but the encoder parameters associated with edge $(5,6)$ is determined greedily first and then other parameters. Such greedy approach doesn't ensure global optimality but provides a general analytical solution (see appendix for further details).

% (see Sec. \ref{subsec:explanation_task_aw} for further details)

The simulation results are shown in Fig. \ref{fig:synthetic_evaluation}. The parameters are as follows: we fix $n = 32$, $\Psi = I$, $a = b \geq 16$, $Z = 8$. Next we let eigen-values $\mu_1^{(3)}, \cdots, \mu_{2Z}^{(3)}$ and $\mu_1^{(4)}, \cdots, \mu_{2Z}^{(4)}$ be positive, and other eigen-values of $S^{(3)}$ and $S^{(4)}$ be 0.
% $\mu_{2Z+1}^{(3)} = \cdots = \mu_{n}^{(3)} = \mu_{2Z+1}^{(4)} = \cdots = \mu_{n}^{(4)} = 0$. 
Hence $\mathcal{L}_{\text{total,lb}} = 0$. Other training details are provided in the appendix.
% Sec. \ref{subsec:training_details}.
In Fig. \ref{fig:synthetic_loss_r34+}, we fix $a=b=24$ and change eigen-vectors $u^{(3)}_{1}, \cdots, u^{(3)}_{2Z}$ and $u^{(4)}_{1}, \cdots, u^{(4)}_{2Z}$ to make $r^{(3,4)}_+$ different, while in the meantime keep Eq. (\ref{eq:sf_1}), (\ref{eq:sf_2}) and (\ref{eq:sf_3}). We can observe our task-aware coding approach achieves $\mathcal{L}_{\text{total,lb}}$ when $r^{(3,4)}_+ \leq 24$, i.e., Eq. (\ref{eq:nc_1}) is satisfied, which verifies our conclusion in Corollary \ref{thm:sf2}. We also notice that the task-aware no coding approach achieves $\mathcal{L}_{\text{total,lb}}$ when $r^{(3,4)}_+ \leq 24$ as well, since coding is not required to achieve $\mathcal{L}_{\text{total,lb}}$. 
In Fig. \ref{fig:synthetic_loss_a}, we fix $u^{(3)}_{1}, \cdots, u^{(3)}_{2Z}$ and $u^{(4)}_{1}, \cdots, u^{(4)}_{2Z}$ such that $r^{(3,4)}_+ = 18$, and change $a$, while in the meantime keep Eq. (\ref{eq:sf_1}) and (\ref{eq:sf_2}). We can observe our task-aware coding approach achieves $\mathcal{L}_{\text{total,lb}}$ when $a = b \geq 24$, i.e., Eq. (\ref{eq:sf_4}) is satisfied, which verifies our conclusion in Corollary \ref{thm:sf3}. Furthermore, in both Fig. \ref{fig:synthetic_loss_r34+} and \ref{fig:synthetic_loss_a}, our task-aware coding approach beats all the other benchmark approaches under varying $r^{(3,4)}_+$'s with respect to overall task loss $\mathcal{L}_{\text{total}}$.

In the next section we will evaluate our approach over four high-dimensional real-world datasets. 

\section{Evaluation}\label{sec:eval}
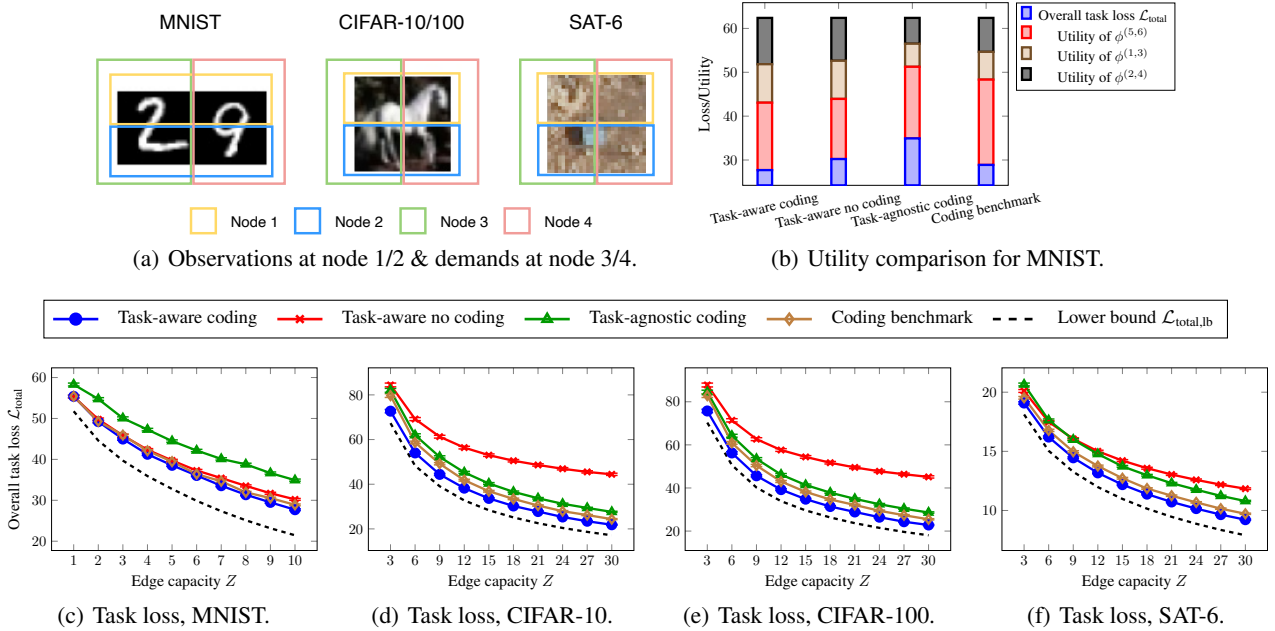
\begin{figure*}[t]
\centering

% \input{figs_tikz/legend_approaches.tex}

% \subfigure{
% \input{figs_tikz/mnist_test_loss_bottleneck.tex}
% \label{fig:loss_bottleneck}
% }
% \subfigure{
% \input{figs_tikz/mnist_test_loss_comparison.tex}
% \label{fig:eval_comparison}
% }
% \begin{minipage}[t]{0.45\linewidth}
% \strut\vspace*{-\baselineskip}\newline
% \subfigure[Observations at node 1/2 \& demands at node 3/4.]{
% \includegraphics[scale=0.75]{figs/coding_visualization.pdf}
% \label{fig:eval_setup}
% }
% \end{minipage}
% \begin{minipage}[t]{0.35\linewidth}
% \subfigure[Utility comparison for MNIST.]{
% \input{figs_tikz/mnist_test_loss_comparison.tex}
% \label{fig:mnist_loss_comparison}
% }
% \end{minipage}

\subfigure[Observations at node 1/2 \& demands at node 3/4.]{
\includegraphics[scale=0.65]{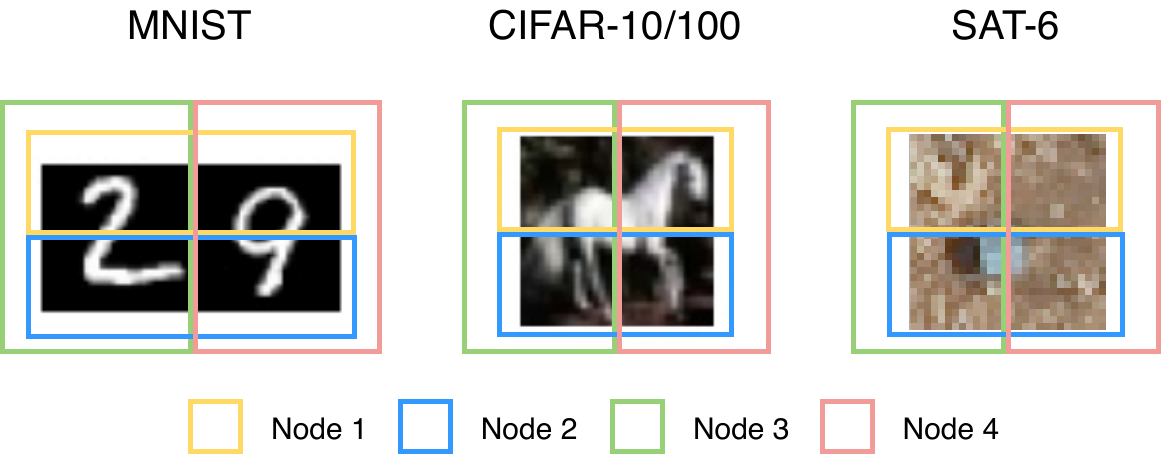}
\label{fig:eval_setup}
}
\subfigure[Utility comparison for MNIST.]{
\begin{tikzpicture}[scale=0.55]
    \begin{axis}[height=6cm, width=8cm,
                 ybar stacked,
                 symbolic x coords={Task-aware coding,
                                    Task-aware no coding,
                                    Task-agnostic coding,
                                    Coding benchmark,},
                 legend pos=outer north east,
                 legend entries = {{} {Overall task loss $\mathcal{L}_{\text{total}}$},
                                   {} {Utility of $\phi^{(5,6)}$},
                                   {} {Utility of $\phi^{(1,3)}$},
                                   {} {Utility of $\phi^{(2,4)}$}},
                 xticklabel style={rotate=12},
                 xtick=data,
                 ylabel={Loss/Utility},
                 xlabel style={font=\large},
                 ylabel style={font=\large},
                 scaled y ticks=false,
                 every axis plot/.append style={ultra thick},
                ]
    \addplot coordinates
        {
(Task-aware coding,27.697919935074044)
(Task-aware no coding,30.21998860678775)
(Task-agnostic coding,34.9290036944808)
(Coding benchmark,28.883398168538903)
        };
    \addplot coordinates
        {
(Task-aware coding,15.3754106800886)
(Task-aware no coding,13.709345788763983)
(Task-agnostic coding,16.32623981221446)
(Coding benchmark,19.446559777400118)
        };
      \addplot coordinates
        {
(Task-aware coding,8.766153624359461)
(Task-aware no coding,8.698099988462484)
(Task-agnostic coding,5.258560854175325)
(Coding benchmark,6.328948331269629)
        };
      \addplot coordinates
        {
(Task-aware coding,10.557859093099275)
(Task-aware no coding,9.769908948607164)
(Task-agnostic coding,5.883538971750799)
(Coding benchmark,7.7384370554127315)
        };
    \end{axis}
\end{tikzpicture}
\label{fig:mnist_loss_comparison}
}

\subfigure{
% \hspace{0.08\linewidth}
\begin{tikzpicture}[scale=0.5]
    \begin{customlegend}[
    legend entries={Task-aware coding,Task-aware no coding,Task-agnostic coding,Coding benchmark,Lower bound $\mathcal{L}_{\text{total,lb}}$},
    legend columns=5,
    legend style={
            /tikz/column 1/.style={
                column sep=5pt,
            },
            /tikz/column 2/.style={
                column sep=5pt,
                font=\scriptsize,
            },
            /tikz/column 3/.style={
                column sep=5pt,
            },
            /tikz/column 4/.style={
                column sep=5pt,
                font=\scriptsize,
            },
            /tikz/column 5/.style={
                column sep=5pt,
            },
            /tikz/column 6/.style={
                column sep=5pt,
                font=\scriptsize,
            },
            /tikz/column 7/.style={
                column sep=5pt,
            },
            /tikz/column 8/.style={
                column sep=5pt,
                font=\scriptsize,
            },
            /tikz/column 9/.style={
                column sep=5pt,
            },
            /tikz/column 10/.style={
                font=\scriptsize,
            },
        }]
    \addlegendimage{blue,mark=*,line width=1pt}
    \addlegendimage{red,mark=x,line width=1pt}
    \addlegendimage{green!60!black,mark=triangle,line width=1pt}
    \addlegendimage{brown,mark=diamond,line width=1pt}
    \addlegendimage{black,dashed, mark=none,line width=1pt}
    \end{customlegend}
\end{tikzpicture}
}\\
\addtocounter{subfigure}{-1}
\subfigure[Task loss, MNIST.]{
\begin{tikzpicture}[scale=0.55]
    \pgfplotsset{normalsize,samples=10}
    \begin{axis}[ height=6cm, width=8cm,
                  % legend style={at={(0.5,0.25)},
                  %              anchor=north,legend columns=3},
                  % legend pos=outer north east,
                  % legend entries = {Task-aware coding,
                  %                   Task-aware no coding,
                  %                   Task-agnostic coding,
                  %                   Coding benchmark,
                  %                   Lower bound $\mathcal{L}_{\text{total,lb}}$},
                  % xmin=0,xmax=1.0,ymin=0,ymax=0.035,
                  xtick distance=1,
                  xlabel={Edge capacity $Z$},
                  ylabel={Overall task loss $\mathcal{L}_{\text{total}}$},
                  xlabel style={font=\large},
                  ylabel style={font=\large},
                  scaled y ticks=false,
                  every axis plot/.append style={ultra thick},
                  mark size=3pt ]
            \addplot [blue,mark=*, error bars/.cd, y dir=both, y explicit,
                      error bar style={line width=2pt,solid},
                      error mark options={line width=1pt,mark size=4pt,rotate=90}]
                      table [x=x, y=y, y error=y-err]{%
                        x y y-err
1 55.39141340352539 0.33202762028089694
2 49.22432985096673 0.31115948766503404
3 44.97046357183471 0.29099564180669957
4 41.25007260826336 0.271000916587609
5 38.547063239500496 0.2672398290055057
6 36.06695440069438 0.25886055895681315
7 33.566568702929644 0.2509456261707066
8 31.36170147405591 0.23775558890146395
9 29.494767170088792 0.22621333647341269
10 27.69791993507404 0.21641226321816942
                      };
            \addplot [red,mark=x, error bars/.cd, y dir=both, y explicit,
                      error bar style={line width=2pt,solid},
                      error mark options={line width=1pt,mark size=4pt,rotate=90}]
                      table [x=x, y=y, y error=y-err]{%
                        x y y-err
1 55.29647523906259 0.3330357878716796
2 49.79136833161074 0.3217153726619483
3 45.83213691948825 0.3012487787578678
4 42.34773623115551 0.28389085735344605
5 39.7362289918937 0.2759520780645349
6 37.27372200317249 0.26957275554827104
7 35.433841989305726 0.25829202628233006
8 33.523513806284406 0.24888047305007938
9 31.723690588212573 0.24222450313551036
10 30.21998860678775 0.23226521358968877
                      };
            \addplot [green!60!black,mark=triangle, error bars/.cd, y dir=both, y explicit,
                      error bar style={line width=2pt,solid},
                      error mark options={line width=1pt,mark size=4pt,rotate=90}]
                      table [x=x, y=y, y error=y-err]{%
                        x y y-err
1 58.295224196336584 0.35332050489757083
2 54.748591897122935 0.3385160274063127
3 50.06340897840224 0.32881115913346576
4 47.28081136829607 0.30530707222853154
5 44.47622751224086 0.2939808071841702
6 42.1663398121924 0.289164430380261
7 40.111165416026324 0.28536894497674054
8 38.784705052885776 0.2801104728435548
9 36.65190408192145 0.26944517763107495
10 34.9290036944808 0.2612232761548838
                      };
            \addplot [brown,mark=diamond, error bars/.cd, y dir=both, y explicit,
                      error bar style={line width=2pt,solid},
                      error mark options={line width=1pt,mark size=4pt,rotate=90}]
                      table [x=x, y=y, y error=y-err]{%
                        x y y-err
1 55.410060542076174 0.33461809050819524
2 49.29517472452035 0.31957396288265894
3 45.92902739076235 0.29303488616050427
4 42.0911240301845 0.27782193741679384
5 39.446030487057804 0.2754469837061677
6 36.39940962810966 0.26190837108070314
7 34.645866314040894 0.25171774324382334
8 31.94894211060525 0.2364500159943698
9 30.5474251041847 0.22565938293532267
10 28.883398168538907 0.2159762622396713
                      };
            \addplot [black,dashed, mark=none, error bars/.cd, y dir=both, y explicit,
                      error bar style={line width=2pt,solid},
                      error mark options={line width=1pt,mark size=4pt,rotate=90}]
                      table [x=x, y=y]{%
                        x y
1 51.75517566750313 0
2 44.5907538664608 0
3 39.70291034926374 0
4 35.9204163418521 0
5 32.76757879558056 0
6 29.869012627356888 0
7 27.34203643259282 0
8 25.116145736264457 0
9 23.16722325991975 0
10 21.449470474082702 0
                      };
    \end{axis}
\end{tikzpicture}
\label{fig:mnist_loss_bottleneck}
}
\subfigure[Task loss, CIFAR-10.]{
\begin{tikzpicture}[scale=0.55]
    \pgfplotsset{normalsize,samples=10}
    \begin{axis}[ height=6cm, width=8cm,
                 % legend pos=outer north east,
                 % legend entries = {{} {Task-aware coding},
                 %                   {} {Task-aware no coding},
                 %                   {} {Task-agnostic coding},
                 %                   {} {Coding benchmark},
                 %                   {} {Lower bound $\mathcal{L}_{\text{total,lb}}$}},
                  xtick distance=3,
                  xlabel={Edge capacity $Z$},
                  xlabel style={font=\large},
                  ylabel style={font=\large},
                  scaled y ticks=false,
                  every axis plot/.append style={ultra thick},
                  mark size=3pt ]
            \addplot [blue,mark=*, error bars/.cd, y dir=both, y explicit,
                      error bar style={line width=2pt,solid},
                      error mark options={line width=1pt,mark size=4pt,rotate=90}]
                      table [x=x, y=y, y error=y-err]{%
                        x y y-err
3 72.71321517861047 0.7376326495461132
6 54.01205546054677 0.530648125847083
9 44.40110008799114 0.4441973167528023
12 38.24219118464862 0.3840795358803665
15 33.64989174098309 0.33986815442102336
18 30.258972485513695 0.30566759498160007
21 27.769647753133903 0.282338478512876
24 25.331620358633273 0.2577422368287276
27 23.428259085478942 0.23990496800833203
30 21.853545578979652 0.22437263026269522
                      };
            \addplot [red,mark=x, error bars/.cd, y dir=both, y explicit,
                      error bar style={line width=2pt,solid},
                      error mark options={line width=1pt,mark size=4pt,rotate=90}]
                      table [x=x, y=y, y error=y-err]{%
                        x y y-err
3 84.44633937117396 0.8430835484684671
6 69.16943815921675 0.685327742505536
9 61.28328924095693 0.6268668028992312
12 56.3884506466035 0.5852123397079898
15 53.01424733963739 0.559270172022624
18 50.52101156753098 0.5399202823013617
21 48.62831816694367 0.5261875168601299
24 46.92258996901775 0.5145296751329119
27 45.521313628253154 0.504791990349663
30 44.4377070977898 0.4979846425741489
                      };
            \addplot [green!60!black,mark=triangle, error bars/.cd, y dir=both, y explicit,
                      error bar style={line width=2pt,solid},
                      error mark options={line width=1pt,mark size=4pt,rotate=90}]
                      table [x=x, y=y, y error=y-err]{%
                        x y y-err
3 82.03903072768894 0.8578395540868841
6 61.964864560823955 0.6240392243227669
9 52.21408844788745 0.5280406212963982
12 45.29641722984952 0.45861645110114196
15 40.32791857172415 0.4051309894251167
18 36.592401071186735 0.36754020609942684
21 33.71281405946831 0.34054161195654636
24 31.24493514818392 0.3154517086446587
27 29.383090142984067 0.29726360823729153
30 27.615985513385223 0.28102797182097805
                      };
            \addplot [brown,mark=diamond, error bars/.cd, y dir=both, y explicit,
                      error bar style={line width=2pt,solid},
                      error mark options={line width=1pt,mark size=4pt,rotate=90}]
                      table [x=x, y=y, y error=y-err]{%
                        x y y-err
3 79.67238445588364 0.8235808742300994
6 58.70653405279929 0.5843203172338522
9 49.17087747799039 0.4981645986068072
12 41.816915906892255 0.42034671868156254
15 36.86040578804735 0.3721553910853388
18 33.42799234683219 0.33837890786284186
21 30.467141467478488 0.30915091875880923
24 27.925306035535172 0.28393023965697006
27 26.196223067652223 0.26750087224862773
30 24.375456475810836 0.24951751713016507
                      };
            \addplot [black,dashed, mark=none, error bars/.cd, y dir=both, y explicit,
                      error bar style={line width=2pt,solid},
                      error mark options={line width=1pt,mark size=4pt,rotate=90}]
                      table [x=x, y=y]{%
                        x y
3 67.3497119794032 0
6 48.26577995707447 0
9 39.08305483730237 0
12 32.733900925182525 0
15 28.365531467607042 0
18 25.13417121774455 0
21 22.55002833062153 0
24 20.420239900460246 0
27 18.66520696151906 0
30 17.14243325251809 0
                      };
    \end{axis}
\end{tikzpicture}
\label{fig:cifar_loss_bottleneck}
}
\subfigure[Task loss, CIFAR-100.]{
\begin{tikzpicture}[scale=0.55]
    \pgfplotsset{normalsize,samples=10}
    \begin{axis}[ height=6cm, width=8cm,
                 % legend pos=outer north east,
                 % legend entries = {{} {Task-aware coding},
                 %                   {} {Task-aware no coding},
                 %                   {} {Task-agnostic coding},
                 %                   {} {Coding benchmark},
                 %                   {} {Lower bound $\mathcal{L}_{\text{total,lb}}$}},
                  xtick distance=3,
                  xlabel={Edge capacity $Z$},
                  xlabel style={font=\large},
                  ylabel style={font=\large},
                  scaled y ticks=false,
                  every axis plot/.append style={ultra thick},
                  mark size=3pt ]
            \addplot [blue,mark=*, error bars/.cd, y dir=both, y explicit,
                      error bar style={line width=2pt,solid},
                      error mark options={line width=1pt,mark size=4pt,rotate=90}]
                      table [x=x, y=y, y error=y-err]{%
                        x y y-err
3 75.72532026391171 0.8006212075735749
6 56.25459992009163 0.5991920682626052
9 45.65566082802817 0.4938472354800813
12 39.23819902874307 0.42846871881348514
15 34.809498017049 0.3849390344778091
18 31.395059429965897 0.3508763766282604
21 28.867329899197816 0.32464744509042404
24 26.440276865583385 0.29921987126304433
27 24.338672096552745 0.2779348940948749
30 22.820227130789963 0.26328215174288755
                      };
            \addplot [red,mark=x, error bars/.cd, y dir=both, y explicit,
                      error bar style={line width=2pt,solid},
                      error mark options={line width=1pt,mark size=4pt,rotate=90}]
                      table [x=x, y=y, y error=y-err]{%
                        x y y-err
3 87.7780954121265 0.896843046407539
6 71.39408856672749 0.7507414865709795
9 62.649730824265504 0.6742116424824344
12 57.60259079943213 0.6288580640252885
15 54.40773184882551 0.6016557127274588
18 51.7175068754367 0.582612903686134
21 49.53091333851989 0.565791589954737
24 47.73838538059628 0.5518406561212041
27 46.37492673104992 0.5407555919722689
30 45.130035099177846 0.5314908351403023
                      };
            \addplot [green!60!black,mark=triangle, error bars/.cd, y dir=both, y explicit,
                      error bar style={line width=2pt,solid},
                      error mark options={line width=1pt,mark size=4pt,rotate=90}]
                      table [x=x, y=y, y error=y-err]{%
                        x y y-err
3 84.4876040594846 0.8853615030011157
6 64.23404765771383 0.6827856081596755
9 53.51667802229789 0.5758100108804233
12 46.20193886163251 0.5004388459880631
15 41.38561348107302 0.4501669975560336
18 37.83327479383213 0.4161260502534695
21 34.96357866043998 0.38706957502243355
24 32.45021118990377 0.3616738017257644
27 30.39536941041916 0.34077220628265076
30 28.610642474839814 0.3229563143996316
                      };
            \addplot [brown,mark=diamond, error bars/.cd, y dir=both, y explicit,
                      error bar style={line width=2pt,solid},
                      error mark options={line width=1pt,mark size=4pt,rotate=90}]
                      table [x=x, y=y, y error=y-err]{%
                        x y y-err
3 82.83456474580701 0.8598084605271638
6 60.936674310605916 0.6336884806953504
9 50.41213634896661 0.5365144118070174
12 43.00652635423813 0.4642087727409736
15 38.17978994534919 0.41719524526643026
18 34.66151418002619 0.3818362810993875
21 32.110476827993 0.3563587135696677
24 29.419992322822928 0.3289314488915918
27 27.35219734696731 0.30827188805671185
30 25.461886865297746 0.2887161017077619
                      };
            \addplot [black,dashed, mark=none, error bars/.cd, y dir=both, y explicit,
                      error bar style={line width=2pt,solid},
                      error mark options={line width=1pt,mark size=4pt,rotate=90}]
                      table [x=x, y=y]{%
                        x y
3 70.34347668724897 0
6 50.627584673493175 0
9 40.20182372436588 0
12 33.86966695181249 0
15 29.587989300353087 0
18 26.328397713218408 0
21 23.661529202378375 0
24 21.48763462581752 0
27 19.613388657406546 0
30 18.048006184246987 0

                      };
    \end{axis}
\end{tikzpicture}
\label{fig:cifar_100_loss_bottleneck}
}
\subfigure[Task loss, SAT-6.]{
\begin{tikzpicture}[scale=0.55]
    \pgfplotsset{normalsize,samples=10}
    \begin{axis}[ height=6cm, width=8cm,
                 % legend pos=outer north east,
                 % legend entries = {{} {Task-aware coding},
                 %                   {} {Task-aware no coding},
                 %                   {} {Task-agnostic coding},
                 %                   {} {Coding benchmark},
                 %                   {} {Lower bound $\mathcal{L}_{\text{total,lb}}$}},
                  xtick distance=3,
                  xlabel={Edge capacity $Z$},
                  xlabel style={font=\large},
                  ylabel style={font=\large},
                  scaled y ticks=false,
                  every axis plot/.append style={ultra thick},
                  mark size=3pt ]
            \addplot [blue,mark=*, error bars/.cd, y dir=both, y explicit,
                      error bar style={line width=2pt,solid},
                      error mark options={line width=1pt,mark size=4pt,rotate=90}]
                      table [x=x, y=y, y error=y-err]{%
                        x y y-err
3 19.085267376902728 0.11175671271055022
6 16.206330909594058 0.09174971117758754
9 14.453680315961766 0.07999596013568949
12 13.178612710663934 0.07315495424006402
15 12.184924728162258 0.06704972434432223
18 11.381979619565923 0.06241780611026324
21 10.720772925858183 0.058571391421939706
24 10.160169133001943 0.05546785381389338
27 9.651565105712864 0.052572978074195136
30 9.22245983774499 0.05020198840034211
                      };
            \addplot [red,mark=x, error bars/.cd, y dir=both, y explicit,
                      error bar style={line width=2pt,solid},
                      error mark options={line width=1pt,mark size=4pt,rotate=90}]
                      table [x=x, y=y, y error=y-err]{%
                        x y y-err
3 20.103230601267256 0.12229015514528245
6 17.487684124127757 0.10347882791018333
9 16.051253345869448 0.0944188568220715
12 15.005554255398108 0.0886422540808796
15 14.21538173990757 0.08408344167180586
18 13.572846353513365 0.0803639878852653
21 13.0416902549548 0.07743728263835796
24 12.5880159203934 0.0749148087141357
27 12.185597647237994 0.07291613764978719
30 11.841439830068278 0.07112795338470157
                      };
            \addplot [green!60!black,mark=triangle, error bars/.cd, y dir=both, y explicit,
                      error bar style={line width=2pt,solid},
                      error mark options={line width=1pt,mark size=4pt,rotate=90}]
                      table [x=x, y=y, y error=y-err]{%
                        x y y-err
3 20.651009272988187 0.12007606752263388
6 17.634190307437564 0.10164687404768745
9 15.989032882665661 0.08989229817736034
12 14.776631207516287 0.08250636711062033
15 13.730257983297482 0.0765814103788548
18 12.979836835157476 0.07186478479914046
21 12.310269939316106 0.06785574112472381
24 11.758653620575943 0.06449787210205121
27 11.235365252838344 0.06149450626186008
30 10.79803536509124 0.05901340116267271
                      };
            \addplot [brown,mark=diamond, error bars/.cd, y dir=both, y explicit,
                      error bar style={line width=2pt,solid},
                      error mark options={line width=1pt,mark size=4pt,rotate=90}]
                      table [x=x, y=y, y error=y-err]{%
                        x y y-err
3 19.526060821351592 0.11129865747951942
6 16.773167019183457 0.09496469824195497
9 14.967432639091646 0.08313457400179557
12 13.758784588438102 0.07552604825581623
15 12.72605987597268 0.07062927312129319
18 11.855401710762813 0.06529331369683038
21 11.248827703839634 0.0617264205989311
24 10.684410607072477 0.05843434612982687
27 10.156101731763744 0.05547270217974485
30 9.721446037819453 0.05306825142857077
                      };
            \addplot [black,dashed, mark=none, error bars/.cd, y dir=both, y explicit,
                      error bar style={line width=2pt,solid},
                      error mark options={line width=1pt,mark size=4pt,rotate=90}]
                      table [x=x, y=y]{%
                        x y
3 18.11528773548548 0
6 15.044443375155907 0
9 13.266905134593316 0
12 11.978194228193729 0
15 10.978177609148219 0
18 10.15355464497057 0
21 9.465813813789197 0
24 8.876510601903448 0
27 8.356580935863974 0
30 7.903620875660607 0
                      };
    \end{axis}
\end{tikzpicture}
\label{fig:sat_loss_bottleneck}
}
% \subfigure[Task loss, STL-10.]{
% \input{figs_tikz/stl_test_loss_bottleneck.tex}
% \label{fig:stl_loss_bottleneck}
% }
% \begin{minipage}[t]{0.65\linewidth}
% \strut\vspace*{-\baselineskip}\newline
% \subfigure{
% % \hspace{0.08\linewidth}
% \input{figs_tikz/legend_approaches_two_row.tex}
% }
% \addtocounter{subfigure}{-1}
% \subfigure[]{
% \input{figs_tikz/mnist_test_loss_bottleneck.tex}
% \label{fig:mnist_loss_bottleneck}
% }
% \subfigure[]{
% \input{figs_tikz/cifar_test_loss_bottleneck.tex}
% \label{fig:cifar_loss_bottleneck}
% }
% \end{minipage}\hfill
% \begin{minipage}[t]{0.35\linewidth}
% \strut\vspace*{-\baselineskip}\newline
% \subfigure[]{
% \input{figs_tikz/mnist_test_loss_comparison.tex}
% \label{fig:mnist_loss_comparison}
% }
% \end{minipage}

\caption{
Evaluation setup and result with MNIST, CIFAR-10, CIFAR-100 and SAT-6 dataset. 
% Evaluation setup and result with MNIST, CIFAR-10, CIFAR-100 and STL-10  dataset. 
% Left and middle: overall task loss $\mathcal{L}_\text{total}$ under different edge capacity $Z$ for MNIST ($\gamma = 0.9$) and CIFAR-10. Right: comparison of the utilities of $\phi^{(5,6)}, \phi^{(1,3)}$ and $\phi^{(2,4)}$ for different approaches (MNIST, $Z=10$ and $\gamma=0.9$).
% Evaluation result with MNIST dataset. Left: overall task loss $\mathcal{L}_\text{total}$ under different edge capacity $Z$ (with $\gamma = 0.9$).  Middle: overall task loss $\mathcal{L}_\text{total}$ under different weight $\gamma$ (with $Z = 10$). Right: comparison of the utilities of $\phi^{(5,6)}, \phi^{(1,3)}$ and $\phi^{(2,4)}$ for different approaches (with $Z=10$ and $\gamma=0.9$).
}
\label{fig:full_evaluation}
\end{figure*}

Our evaluation compares the performance of our task-ware coding approach and other benchmark approaches 
(as in the last section)
% (as in Sec. \ref{sec:algorithm}) 
over a few standard ML datasets, including MNIST \citep{lecun1998gradient}, CIFAR-10, CIFAR-100 \citep{krizhevsky2009learning} and SAT-6 \citep{basu2015deepsat}. For MNIST, each data sample is a $28 \times 28$ handwritten digit image, and we let $x$ be a horizontally-concatenated image ($28 \times 56$) of two images.
% , i.e., $n = 1568$. 
Node 1 and 2 observe the upper and the lower half part of the concatenated image (both $14 \times 56$) respectively.
% , i.e., $a = b = 784$
Task matrices $K^{(3)}$ and $K^{(4)}$ are formulated as follows: we pre-train a convolutional neural network (CNN) to classify original MNIST digits by their labels. Task matrix $K^{(3)}$ requires both the reconstruction of the feature map (i.e., the output of the first layer of CNN) of the left MNIST digit in the concatenated image, and the reconstruction of the concatenated image itself. Mathematically, we have
\begin{align}
K^{(3)} = [\underbrace{\gamma \tilde{K}^{(3)\top}}_{\text{recon. of left feature map}}, \underbrace{(1-\gamma)I}_{\text{recon. of conctenated image}}]^\top,
\end{align}
where $\tilde{K}^{(3)}$ represents the mapping between $x$ and the feature map of the left MNIST digit, and $\gamma$ is a weight coefficient. Here we use $\gamma = 0.9$. Task matrix $K^{(4)}$ is formulated similarly while the feature map of the right MNIST digit is considered instead. For CIFAR-10/CIFAR-100/SAT-6, each data sample is a $32 \times 32$ or $28 \times 28$ colored image with 3 or 4 channels and we let $x$ represent the original image. We similarly let node 1 and 2 observe the upper and the lower half part of the image respectively, and let node 3 and node 4 require the reconstruction of the left and the right half part respectively. The setup is illustrated in Fig. \ref{fig:eval_setup}. Other training details are provided in the appendix.
% Sec. \ref{subsec:training_details}.

% The evaluation result is shown in Fig. \ref{fig:full_evaluation}. In Fig. \ref{fig:mnist_loss_bottleneck} and Fig. \ref{fig:cifar_loss_bottleneck}, we plot the overall task loss $\mathcal{L}_{\text{total}}$ under different edge capacity $Z$.

The evaluation result is shown in Fig. \ref{fig:full_evaluation}. In Fig. \ref{fig:mnist_loss_bottleneck}-\ref{fig:sat_loss_bottleneck}, we plot the overall task loss $\mathcal{L}_{\text{total}}$ under different edge capacity $Z$. 
% (with $\gamma = 0.9$)
In these figures, we see task-aware coding and coding benchmark approach outperform task-aware no coding and task-agnostic coding approach, and the overall task loss $\mathcal{L}_{\text{total}}$ of our task-aware coding approach is the closet to $\mathcal{L}_{\text{total, lb}}$. The maximum improvements of overall task loss $\mathcal{L}_{\text{total}}$ for task-aware coding approach are 26.1\%, 26.4\%, 25.3\% and 17.1\% respectively, compared to task-agnostic coding approach; and are 9.1\%, 103.3\%, 97.8\% and 28.4\% respectively, compared to task-aware no coding approach.
We also notice that, task-agnostic coding approach doesn't always outperform task-aware no coding approach, and vice versa. 
% Moreover, we notice task-aware no coding, though being a task-aware approach, even underperforms task-agnostic coding in Fig. \ref{fig:cifar_loss_bottleneck}. 
Therefore, it is beneficial to combine network coding and task-awareness.
% These observations illustrate the benefit of the network coding as well as making it task-aware.

In Fig. \ref{fig:mnist_loss_comparison}, we compare the utilities of $\phi^{(5,6)}, \phi^{(1,3)}$ and $\phi^{(2,4)}$ in terms of minimizing the overall task loss $\mathcal{L}_{\text{total}}$ when $Z = 10$, $\gamma = 0.9$. The three utilities are defined in a way such that their sums plus $\mathcal{L}_{\text{total}}$ is a fixed number (see appendix for the formal definition). 
% as in Sec. \ref{subsec:mnist_utility}.
We observe that the coding benchmark approach outperforms other approaches with respect to the utility of $\phi^{(5,6)}$, but underperforms our task-aware coding approach by 4.2\% with respect to the overall task loss $\mathcal{L}_{\text{total}}$. This is because coding benchmark approach greedily determines the encoder parameters associated with edge $(5, 6)$ first which however could not guarantee optimality. On the other hand, our task-aware coding approach tunes all the encoding and decoding parameters jointly and achieves a lower $\mathcal{L}_{\text{total}}$.

\textbf{Limitation.} Our task-aware network coding problem defined in Eq. (\ref{eq:prob}) is non-convex, and hence the adopted gradient descent method may converge to local optimum. 

\section{Conclusion}\label{sec:conclusion}
This paper considers task-aware network coding over butterfly network in real-coordinate space. We prove a lower bound $\mathcal{L}_{\text{total,lb}}$ of the total loss, as well as conditions for achieving $\mathcal{L}_{\text{total,lb}}$. We also provide a machine learning algorithm in the general settings. Experimental results demonstrate that our task-aware coding approach outperforms the benchmark approaches under various settings.

% shorter version
% Our future work will focus on extending the analysis of the task-aware network coding problem to other network topologies, as well as exploring the possibility of determining the optimal overall task loss in the general case analytically rather than numerically.

%There are two directions along which we can extend this study. 

Regarding future extension, although butterfly network is a representative topology in network coding, it is worthwhile to extend the analysis of the task-aware network coding problem to general networks. A similar $\mathcal{L}_{\text{total,lb}}$ can still be derived, yet the associated necessary condition and sufficient conditions for achieving $\mathcal{L}_{\text{total,lb}}$ depend on the specific network topology in a manner that needs further work to be fully understood.

%Further work is needed to understand how network topology 
%and may require case-by-case derivation. %Second, we will also focus on exploring the possibility of solving the task-aware network coding problem analytically in the general case rather than resorting to machine learning algorithm.

% \newpage
\bibliography{ref}

\newpage
% \appendix
\section{Appendix}
\subsection{Task-aware PCA Derivation} \label{subsec:pca_derivation}

We define random variable $h = L^{-1} x \in \reals^n$ and $\hat{h} = L^{-1} \hat{x} \in \reals^n$. Thus $\E_h [hh^\top] = \E_x [L^{-1} xx^\top L^{-\top}] = I$. And we also define $D_h = L^{-1} D$ and $E_h = E L$.

Then we have
\begin{align}
\mathcal{L} 
= & \E_x [\|K (x - \hat{x}) \|_2^2] = \E_h [\|KL(h-\hat{h})\|_2^2]\notag\\
= & \E_h [\|KL (I - D_h E_h) h \|_2^2]\notag\\
= & \E_h [\tr(KL (I - D_h E_h) hh^\top (I - D_h E_h)^\top L^\top K^\top)]\notag\\
= & \tr(KL (I - D_h E_h) (I - D_h E_h)^\top L^\top K^\top)\label{eq:pca_loss_1}
\end{align}
where $\tr(\cdot)$ denotes the trace of a matrix. We can verify that $D_h = E_h^\top (E_h E_h^\top)^{-1}$ is a zero point of
\begin{align}
\frac{\nabla \mathcal{L}}{\nabla D_h} = 2 L^\top K^\top K L (D_h E_h - I) E_h^\top.
\end{align}
So we plug $D_h = E_h (E_h E_h^\top)^{-1}$ into Eq. \ref{eq:pca_loss_1} and get
\begin{align}
\mathcal{L} = \tr(L^\top K^\top K L) - \tr(L^\top K^\top K L E_h^\top (E_h E_h^\top)^{-1} E_h).\notag
\end{align}
Therefore, we should find $E_h$ that maximizes $\tr(L^\top K^\top K L E_h^\top (E_h E_h^\top)^{-1} E_h)$. Classical PCA dictates that $\tr(L^\top K^\top K L E_h^\top (E_h E_h^\top)^{-1} E_h) \leq \sum_{i=1}^Z \mu_i$, i.e., $\mathcal{L} \geq \sum_{i=Z+1}^n \mu_i$. And if the eigen-gap $\mu_{Z} - \mu_{Z+1} > 0$, the equality holds if and only if $\cl(E_h^\top) = \sn(\{u_1, u_2, \cdots, u_Z\})$, which is equivalent to $\cl(E^\top) = \sn(\{L^{-\top} u_{1}, L^{-\top} u_{2}, \cdots, L^{-\top} u_{Z}\})$.

\subsection{Proof of Theorem \ref{thm:sf0} when $2Z > n$} \label{subsec:proof_thm_sf0}

\begin{proof}
When $2Z > n$, we have $r^{(3,4)}_+ = r^{(3,4)}_- = n$, and $r_-^{(1,3)} = a \geq n - Z$, $r_-^{(2,4)} = b \geq n - Z$. Therefore we aim to make $\cl([\Phi^{(1,3)}, \Phi^{(5,6)}]) = \cl([\Phi^{(2,4)}, \Phi^{(5,6)}]) = \reals^n$. Since butterfly network is symmetric, without loss of generality, we assume $a \leq b$. And we denote the $i$-th column vector of $L$ by $l_i$.

We are able to find $\Phi^{(1,3)}$, $\Phi^{(2,4)}$ and $\Phi^{(5,6)}$ such that $\mathcal{L}_{\text{total,lb}}$ is achievable, through the following three steps:

$\bullet$ i) Let the first $n-b$ column vectors of $\Phi^{(5,6)}$ be $l_1 + l_{a+1}, l_2 + l_{a+2}, \cdots, l_{n-b} + l_{a+n-b}$, and the next $b-a$ column vectors of $\Phi^{(5,6)}$ be $l_{a+n-b+1}, l_{a+n-b+2}, \cdots, l_{n}$; 
% There are $Z-n+a$ column vectors of $\Phi^{(6,3)} = \Phi^{(6,4)}$ to be determined;

$\bullet$ ii) Let the first $n-b$ column vectors of $\Phi^{(1,3)}$ and $\Phi^{(2,4)}$ be $l_1, l_2, \cdots, l_{n-b}$ and $l_{a+1}, l_{a+2}, \cdots, l_{a+n-b}$ respectively; 
% There are $Z-n+b$ column vectors of $\Phi^{(1,3)}$ and $\Phi^{(2,4)}$ to be determined;

$\bullet$ iii) Since $2Z > n$, we have $Z-n+a + Z-n+b > a + b - n$, i.e., the number of non-determined column vectors of $\Phi^{(5,6)}$ and $\Phi^{(1,3)}$ together is greater than the number of remaining vectors $l_{n-b+1}, l_{n-b+2}, \cdots, l_{a}$. So we can assign $l_{n-b+1}, l_{n-b+2}, \cdots, l_{a}$ to the non-determined column vectors of $\Phi^{(5,6)}$ and $\Phi^{(1,3)}$ such that each vector among $l_{n-b+1}, l_{n-b+2}, \cdots, l_{a}$ is assigned at least once. In the end we assign the last $Z-n+b$ column vectors of $\Phi^{(1,3)}$ to the last $Z-n+b$ column vectors of $\Phi^{(2,4)}$.

It can be verified that $\cl([\Phi^{(1,3)}, \Phi^{(5,6)}]) = \cl([\Phi^{(2,4)}, \Phi^{(5,6)}]) = \cl(L) = \reals^n$. Therefore, $\mathcal{L}_{\text{total,lb}}$ is achievable.
\end{proof}

\subsection{Further explanation of the task-aware coding benchmark approach} \label{subsec:explanation_task_aw}

In the coding benchmark approach, our first step is to let the column vectors of $\Phi^{(5,6)}$ be the top-$Z$ normalized eigen-vectors of $S^{(3)} + S^{(4)}$. This greedy step ensures $\mathcal{L}_{\text{total}}$ is minimized when node 3 and 4 receive $\phi^{(5,6)}$ only. 

Next, to determine optimal $\Phi^{(1,3)}$ which further minimizes $\mathcal{L}_{\text{total}}$ as much as possible, we consider the following problem:
\begin{align}
\max_{\tilde{\Phi}^{(1,3)} \in \reals^{n \times \min\{Z, c\}}} \quad & \tilde{\Phi}^{(1,3)\top} S^{(3)} \tilde{\Phi}^{(1,3)}\\
\text{s.t.} \quad & \cl(\tilde{\Phi}^{(1,3)}) \subseteq V,\\
&\tilde{\Phi}^{(1,3)\top}\tilde{\Phi}^{(1,3)} = I,
\end{align}
where $V = \{v \in \cl([\Phi^{(5,6)}, U^{(1)}])| v \perp w, \forall w \in \cl(\Phi^{(5,6)})\}$, and $c = \dm(\cl([\Phi^{(5,6)}, U^{(1)}]))-Z$ is the dimension of $V$. The following $\tilde{\Phi}^{(1,3)}$ is an optimal solution for the considered problem: we formulate a matrix $W \in \reals^{n \times c}$ whose column vectors form an orthogonal basis of $V$, and then let the column vectors of $\tilde{\Phi}^{(1,3)}$ be $W$ times the top-$\min\{Z, c\}$ normalized eigen-vectors of $W^\top S^{(3)} W$. For fixed $\Phi^{(5,6)}$, we know the optimal $\Phi^{(1,3)}$ that minimizes $\mathcal{L}_{\text{total}}$ should satisfy $\cl([\Phi^{(1,3)}, \Phi^{(5,6)}]) = \cl([\tilde{\Phi}^{(1,3)}, \Phi^{(5,6)}])$. Yet in general $\tilde{\Phi}^{(1,3)} \nsubseteq U^{(1)}$, so we cannot assign the column vectors of $\tilde{\Phi}^{(1,3)}$ to the column vectors of $\Phi^{(1,3)}$ directly. Therefore, we let the first $\min\{Z, c\}$ column vectors of $\Phi^{(1,3)}$ be the vectors in $\cl(U^{(1)})$ that extend a basis of $\cl(\Phi^{(5,6)})$ to a basis of $\cl([\tilde{\Phi}^{(1,3)}, \Phi^{(5,6)}])$. Other non-determined column vectors of $\Phi^{(1,3)}$, if any, can be arbitrary vectors in $\cl(U^{(1)})$.

We also use the similar idea to determine the column vectors of $\Phi^{(2,4)}$.

\subsection{Training Details} \label{subsec:training_details}

% \begin{table*}[ht]
% \caption{Training Details}
% \label{tab:train_detail}
% \vskip 0.15in
% \begin{center}
% \begin{tabular}{lccccr}
% \toprule
% Application & Num of Samples (Train/Test) & Num of Epochs & Batch Size & Learning Rate & Runtime\\
% \midrule
% Simulation (Sec. \ref{sec:algorithm}) & 64/64 & 2000 & 64& $0.05$ & $<24$ hrs\\
% MNIST (Sec. \ref{sec:eval}) & 30000/5000 & 20 & 64& $5\times10^{-3}$ & $<24$ hrs\\
% CIFAR-10 (Sec. \ref{sec:eval}) & 50000/10000 & 20 & 64& $5\times10^{-3}$ & $<24$ hrs\\
% CIFAR-100 (Sec. \ref{sec:eval}) & 50000/10000 & 20 & 64& $5\times10^{-3}$ & $<24$ hrs\\
% SAT-6 (Sec. \ref{sec:eval}) & 324000/81000 & 5 & 64& $5\times10^{-3}$ & $<72$ hrs\\
% % STL-10 (Sec. \ref{sec:eval}) & 5000/8000 & 10 & 64& $5\times10^{-3}$ & $<48$ hrs\\
% \bottomrule
% \end{tabular}
% \end{center}
% \vskip -0.1in
% \end{table*}

\begin{table*}[ht]
\caption{Training Details}
\label{tab:train_detail}
% \vskip 0.15in
\begin{center}
\begin{tabular}{lccccr}
\toprule
Application & Num of Samples (Train/Test) & Num of Epochs & Batch Size & Learning Rate & Runtime\\
\midrule
Simulation & 64/64 & 2000 & 64& $0.05$ & $<24$ hrs\\
MNIST & 30000/5000 & 20 & 64& $5\times10^{-3}$ & $<24$ hrs\\
CIFAR-10 & 50000/10000 & 20 & 64& $5\times10^{-3}$ & $<24$ hrs\\
CIFAR-100 & 50000/10000 & 20 & 64& $5\times10^{-3}$ & $<24$ hrs\\
SAT-6 & 324000/81000 & 5 & 64& $5\times10^{-3}$ & $<72$ hrs\\
% STL-10 (Sec. \ref{sec:eval}) & 5000/8000 & 10 & 64& $5\times10^{-3}$ & $<48$ hrs\\
\bottomrule
\end{tabular}
\end{center}
% \vskip -0.1in
\end{table*}

Both our simulation and our evaluation run on a personal laptop with 2.7 GHz Intel Core I5 processor and 8-GB 1867 MHz DDR3 memory. Our code is based on Pytorch and the Adam optimizer is used. The number of samples (train/test), number of epochs, batch size, learning rate and corresponding runtime (for the whole experiment) are summarized in Table \ref{tab:train_detail}. Note that our simulation is based on synthetic data, so we only need a training/testing dataset that satisfies $\E_x [x] = \bm{0}, \Psi= I$ and hence $2n = 64$ samples are enough. The publicly-available MNIST, CIFAR-10, CIFAR-100 and SAT-6 datasets do not have a stated license online.

% Both our simulation in Sec. \ref{sec:algorithm} and our evaluation in Sec. \ref{sec:eval} run on a personal laptop with 2.7 GHz Intel Core I5 processor and 8-GB 1867 MHz DDR3 memory. Our code is based on Pytorch and the Adam optimizer is used. The number of samples (train/test), number of epochs, batch size, learning rate and corresponding runtime (for the whole experiment) are summarized in Table \ref{tab:train_detail}. Note that our simulation in Sec. \ref{sec:algorithm} is based on synthetic data, so we only need a training/testing dataset that satisfies $\E_x [x] = \bm{0}, \Psi= I$ and hence $2n = 64$ samples are enough. The publicly-available MNIST, CIFAR-10, CIFAR-100 and SAT-6 datasets do not have a stated license online.

Moreover, our evaluation for MNIST also requires a pre-trained CNN classifier for $28\times28$ images, which is composed of two consecutive convolution layers and a final linear layer. The number of input channels, the number of output channels, kernel size, stride and padding for two convolution layers are 1, 16, 5, 1, 2 and 16, 32, 5, 1, 2 respectively, and ReLU activation and max pooling with kernel size 2 are used after each convolution layer. The final linear layer has input size 1568 and output size 1.

% Moreover, our evaluation for MNIST in Sec. \ref{sec:eval} also requires a pre-trained CNN classifier for $28\times28$ images, which is composed of two consecutive convolution layers and a final linear layer. The number of input channels, the number of output channels, kernel size, stride and padding for two convolution layers are 1, 16, 5, 1, 2 and 16, 32, 5, 1, 2 respectively, and ReLU activation and max pooling with kernel size 2 are used after each convolution layer. The final linear layer has input size 1568 and output size 1.

% \textbf{We also provide all our documented code as supplementary material and will make it publicly-available after the review process}.

We also provide our code in the supplementary material and will make it publicly-available after the review process.

\subsection{Overall Task Losses of Task-agnostic Coding Approach for Simulation} \label{subsec:simulation_task_ag}

\begin{figure}[t]
\centering

\begin{tikzpicture}[scale=0.5]
    \begin{customlegend}[
    legend entries={Task-aware coding,Task-agnostic coding},
    legend columns=2,
    legend style={
            /tikz/column 1/.style={
                column sep=5pt,
            },
            /tikz/column 2/.style={
                column sep=5pt,
                font=\scriptsize,
            },
            /tikz/column 3/.style={
                column sep=5pt,
            },
            /tikz/column 4/.style={
                column sep=5pt,
                font=\scriptsize,
            },
        }]
    \addlegendimage{blue,mark=*,line width=1pt}
    \addlegendimage{green!60!black,mark=triangle,line width=1pt}
    \end{customlegend}
\end{tikzpicture}

\begin{minipage}[t]{0.49\linewidth}
\strut\vspace*{-\baselineskip}\newline
\subfigure[]{
\begin{tikzpicture}[scale=0.53]
    \pgfplotsset{normalsize,samples=10}
    \begin{axis}[ height=6cm, width=8cm,
                  % legend style={at={(0.5,0.25)},
                  %              anchor=north,legend columns=3},
                  % legend pos=outer north east,
                  % legend entries = {task aware coding,
                  %                   task aware no coding,
                  %                   task agnostic coding,
                  %                   benchmark,
                  %                   lower bound},
                  % xmin=0,xmax=1.0,ymin=0,ymax=0.035,
                  % xtick distance=1,
                  xlabel={$r^{(3,4)}_+$},
                  ylabel={Overall task loss $\mathcal{L}_{\text{total}}$},
                  xlabel style={font=\Large},
                  ylabel style={font=\large},
                  scaled y ticks=false,
                  every axis plot/.append style={ultra thick},
                  mark size=3pt ]
            \addplot [blue,mark=*, error bars/.cd, y dir=both, y explicit,
                      error bar style={line width=2pt,solid},
                      error mark options={line width=1pt,mark size=4pt,rotate=90}]
                      table [x=x, y=y, y error=y-err]{%
                        x y y-err
16 6.911812548195489e-07 3.073093965247263e-07
17 5.066640246184788e-06 4.731035325567704e-06
18 1.4494722576517923e-08 8.568561623689385e-09
19 3.7463711480241007e-07 5.220257716285899e-07
20 2.449196049536813e-07 3.004506629594412e-07
21 1.040607480574976e-08 5.747604562830079e-09
22 4.4599350823704945e-08 2.4596522828269098e-08
23 2.4981650065514156e-07 2.1498828483884266e-07
24 0.003985947765315031 0.0018782044144822862
25 0.6287877687241475 0.3296177718246637
26 1.3985435289708184 0.5970179153308545
27 2.0173303157938416 0.5896322900229524
28 2.7016288557100814 0.7252355307733556
29 3.2574517610267923 0.666716962805083
30 3.950915976422186 0.8473061459132485
31 4.470203985259172 0.7335067487170107
32 5.222534629085299 0.6334240631758826
                      };
%             \addplot [red,mark=x, error bars/.cd, y dir=both, y explicit,
%                       error bar style={line width=2pt,solid},
%                       error mark options={line width=1pt,mark size=4pt,rotate=90}]
%                       table [x=x, y=y, y error=y-err]{%
%                         x y y-err
% 16 1.742752311752271e-29 5.730866711219217e-30
% 17 2.2890390381711667e-06 3.208716095576786e-06
% 18 1.2096275430612091e-13 1.0219516981021993e-13
% 19 0.005503899890213851 0.005503531892064301
% 20 1.938916179175619e-09 2.4857361113876252e-09
% 21 1.211837155947631e-06 1.0760555426549954e-06
% 22 4.933278383830637e-06 3.3127497990603083e-06
% 23 4.90165416283887e-10 3.429956141126788e-10
% 24 6.546608913181524e-05 4.8794794796865374e-05
% 25 0.6649857977148086 0.35021088295879266
% 26 1.439081132176372 0.6905948150178822
% 27 2.088807174345342 0.6198602793712682
% 28 2.7754194854035577 0.7106754686004189
% 29 3.4392795633906537 0.7687401368030804
% 30 4.062084090001905 0.8930110456623133
% 31 4.679884050573273 0.7165165233855438
% 32 5.2653395974012 0.8237096580399346
%                       };
            \addplot [green!60!black,mark=triangle, error bars/.cd, y dir=both, y explicit,
                      error bar style={line width=2pt,solid},
                      error mark options={line width=1pt,mark size=4pt,rotate=90}]
                      table [x=x, y=y, y error=y-err]{%
                        x y y-err
16 172.8375987897491 47.724847121914344
17 94.26239487799364 19.671373634057638
18 174.18427743692456 28.564440985114015
19 153.83579627635987 44.25555735925807
20 130.1169557970644 22.319986174234273
21 158.7624182280638 34.454522365565204
22 149.8457812863585 32.44706573459772
23 156.46833736789608 31.478330604830283
24 148.00023160061866 28.802687306532757
25 121.31254481189269 19.927416701892422
26 154.0859798561758 26.848979893047566
27 130.38710495072678 23.84811046094067
28 168.49931352001445 31.56676558030003
29 163.96303725166513 42.40842934528732
30 148.68978175424812 30.336870944823954
31 156.41032519017318 27.607111116237732
32 103.32288648321963 17.970249843232764
                      };
%             \addplot [brown,mark=diamond, error bars/.cd, y dir=both, y explicit,
%                       error bar style={line width=2pt,solid},
%                       error mark options={line width=1pt,mark size=4pt,rotate=90}]
%                       table [x=x, y=y, y error=y-err]{%
%                         x y y-err
% 16 1.5833892589714623e-29 4.635687923215124e-30
% 17 1.1260613734286689 0.38288194039503937
% 18 1.3631532064732141 0.44691973668179164
% 19 1.6841791041242855 0.48955672260661887
% 20 2.4659839428046615 0.7137615365849697
% 21 2.9211683011400527 0.7623771591516395
% 22 3.202657351318381 0.7498832656669853
% 23 3.486696320035632 0.7257083609174176
% 24 3.7184039463777374 0.6532013575787966
% 25 4.407380780024825 0.8474760654931122
% 26 5.413272170163493 0.9942789025141056
% 27 4.792027967436242 0.8069740159750458
% 28 5.069207318588853 0.7418878368767738
% 29 5.385768487248315 0.824117134123108
% 30 5.7433438333280336 0.8603413886063387
% 31 5.4370998900855 0.7808543963950577
% 32 5.265337455493307 0.8237117375582347
%                       };
            \addplot [black,dashed, mark=none, error bars/.cd, y dir=both, y explicit,
                      error bar style={line width=2pt,solid},
                      error mark options={line width=1pt,mark size=4pt,rotate=90}]
                      table [x=x, y=y]{%
                        x y
16 0
32 0
                      };
    \end{axis}
\end{tikzpicture}
\label{fig:synthetic_loss_r34+_appen}
}
\end{minipage}
\begin{minipage}[t]{0.49\linewidth}
\strut\vspace*{-\baselineskip}\newline
\subfigure[]{
\begin{tikzpicture}[scale=0.53]
    \pgfplotsset{normalsize,samples=10}
    \begin{axis}[ height=6cm, width=8cm,
                  % legend style={at={(0.5,0.25)},
                  %              anchor=north,legend columns=3},
                  % legend pos=outer north east,
                  % legend entries = {task aware coding,
                  %                   task aware no coding,
                  %                   task agnostic coding,
                  %                   benchmark,
                  %                   lower bound},
                  % xmin=0,xmax=1.0,ymin=0,ymax=0.035,
                  % xtick distance=1,
                  xlabel={$a (=b)$},
                  xlabel style={font=\Large},
                  ylabel style={font=\large},
                  scaled y ticks=false,
                  every axis plot/.append style={ultra thick},
                  mark size=3pt ]
            \addplot [blue,mark=*, error bars/.cd, y dir=both, y explicit,
                      error bar style={line width=2pt,solid},
                      error mark options={line width=1pt,mark size=4pt,rotate=90}]
                      table [x=x, y=y, y error=y-err]{%
                        x y y-err
16 6.7887583596120065 1.3249825957190458
17 6.081541798097357 1.329619261431759
18 5.471139056021478 1.389162055664233
19 3.7577250566057714 1.1256888558544584
20 2.884418669358931 1.0126211404453778
21 1.8355849644996638 0.8183213708915495
22 0.8821168576103318 0.3862255576440212
23 0.4083817648606587 0.2461480301702073
24 1.067947262089434e-11 2.5716226398142827e-12
25 9.46066349056175e-12 1.6935688272174662e-12
26 1.1577685293234012e-11 2.4185845724366454e-12
27 7.709805233056714e-08 1.0003009778788454e-07
28 1.417006439065208e-11 3.3575133693265813e-12
29 1.3213822034014166e-11 7.773963355986723e-12
30 7.817713672374309e-12 1.7577774127902056e-12
31 9.137685871359337e-12 2.6289755484753194e-12
32 9.081275297734164e-12 1.733473704279905e-12
                      };
%             \addplot [red,mark=x, error bars/.cd, y dir=both, y explicit,
%                       error bar style={line width=2pt,solid},
%                       error mark options={line width=1pt,mark size=4pt,rotate=90}]
%                       table [x=x, y=y, y error=y-err]{%
%                         x y y-err
% 16 37.95434945235194 9.226613349288886
% 17 34.44220156560836 8.794525333239863
% 18 31.11666840229244 8.347713724698014
% 19 27.361329927213163 8.255567990106641
% 20 24.732487391550794 8.032952735060224
% 21 20.33029476627106 7.352346514783946
% 22 16.80647585981944 7.028196663272082
% 23 14.701982840256784 6.918209261016229
% 24 11.03403817465439 5.921611040509638
% 25 9.513402956501773 5.690230502401569
% 26 7.5491721214631085 5.165000834939477
% 27 6.111124602561819 4.871116008134628
% 28 3.910744764790789 3.9121252516009855
% 29 1.5215836375561946 2.132926536415867
% 30 3.07216461573355e-07 4.306223952321878e-07
% 31 1.4677685774388845e-08 2.0412965465358896e-08
% 32 0.0005916304464809701 0.00018367132264438609
%                       };
            \addplot [green!60!black,mark=triangle, error bars/.cd, y dir=both, y explicit,
                      error bar style={line width=2pt,solid},
                      error mark options={line width=1pt,mark size=4pt,rotate=90}]
                      table [x=x, y=y, y error=y-err]{%
                        x y y-err
16 180.94475147145315 32.67183961020032
17 253.3906242465415 88.10778236802467
18 326.2451391348999 206.64532719414515
19 192.4431457469351 65.41601142742745
20 259.0869113563947 66.85752102100702
21 221.89017561068948 74.34429999873869
22 380.6332677204433 225.33149433141952
23 293.023320169583 97.74720100285599
24 303.7692319400638 107.56738477163343
25 235.07202106690488 63.49304675012292
26 307.9263476944641 101.79277593432704
27 295.70907658554927 118.40815984547349
28 267.469221296904 76.5472062421604
29 319.0429011516379 165.94029539651308
30 322.08980821016564 121.05731711122395
31 324.983549714093 132.45662438199093
32 337.51753378774777 172.94070340715322
                      };
%             \addplot [brown,mark=diamond, error bars/.cd, y dir=both, y explicit,
%                       error bar style={line width=2pt,solid},
%                       error mark options={line width=1pt,mark size=4pt,rotate=90}]
%                       table [x=x, y=y, y error=y-err]{%
%                         x y y-err
% 16 18.486370519701005 2.554308048686343
% 17 17.566420932595655 2.646875405570866
% 18 16.786696249531595 2.73724470998191
% 19 15.373021622232717 3.1926975733545393
% 20 13.646280515199452 3.1634653399850095
% 21 10.943880382314727 2.7131735754466595
% 22 9.674201116063395 2.523222019557381
% 23 7.943644220784298 2.053148526004178
% 24 5.540677786405476 1.3818451471976547
% 25 4.829952306703767 1.242692871758804
% 26 4.070351294157554 1.6885254549526247
% 27 4.070351294157554 1.6885254549526243
% 28 4.070351294157554 1.6885254549526238
% 29 4.070351294157554 1.6885254549526243
% 30 4.070351294157554 1.6885254549526214
% 31 4.070351294157554 1.688525454952624
% 32 4.070351294157554 1.688525454952625
%                       };
            \addplot [black,dashed, mark=none, error bars/.cd, y dir=both, y explicit,
                      error bar style={line width=2pt,solid},
                      error mark options={line width=1pt,mark size=4pt,rotate=90}]
                      table [x=x, y=y]{%
                        x y
16 0
32 0
                      };
    \end{axis}
\end{tikzpicture}
\label{fig:synthetic_loss_a_appen}
}
\end{minipage}

\caption{Simulation result with synthetic data of task-agnostic coding approach: overall task loss $\mathcal{L}_\text{total}$ under different $r_+^{(3,4)}$ (left) and different $a$ (right). }
\label{fig:synthetic_evaluation_appen}
\end{figure}
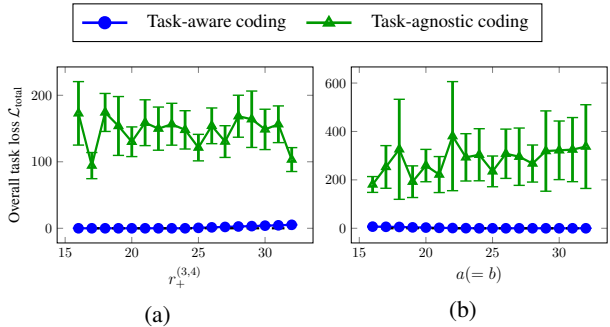

Fig. \ref{fig:synthetic_evaluation_appen} shows the overall task losses of task-agnostic coding approach for the two simulations. For the purpose of comparison we also plot the overall task losses of our task-aware coding approach (same values as in Fig. \ref{fig:synthetic_evaluation}). It clearly illustrates that task-agnostic coding approach performs more poorly than the task-aware coding approach.

\subsection{Definitions of the Utilities in Evaluation} \label{subsec:mnist_utility}

We first define the utility of $\phi^{(5,6)}$. We can find normalized vectors $\xi_1, \cdots, \xi_d$ that form an \textit{orthogonal} basis of $\cl(\Phi^{(5,6)})$, where $d = \dm(\cl(\Phi^{(5,6)}))$. Then the utility of $\phi^{(5,6)}$ is defined as $\sum_{j=1}^d \xi_j^\top (S^{(3)} + S^{(4)}) \xi_j$, which is essentially the minimum achievable overall task loss difference between the case when node 3 and 4 receive nothing and when they receive $\phi^{(5,6)}$ only. For task-aware coding benchmark approach, according to \ref{subsec:explanation_task_aw}, the utility of $\phi^{(5,6)}$ equals the sum of the top-$Z$ eigen-values of $S^{(3)} + S^{(4)}$.

We next define the utility of $\phi^{(1,3)}$. We can find normalized vectors $\chi_1, \cdots, \chi_e$ that extend $\xi_1, \cdots, \xi_d$ to an \textit{orthogonal} basis of $\cl([\Phi^{(1,3)}, \Phi^{(5,6)}])$, where $e = \dm(\cl([\Phi^{(1,3)}, \Phi^{(5,6)}])) - d$. Then the utility of $\phi^{(1,3)}$ is defined as $\sum_{j=1}^e \chi_j^\top S^{(3)} \chi_j$, which is essentially the minimum achievable overall task loss difference between the case when node 3 receives $\phi^{(5,6)}$ only and when it receives both $\phi^{(1,3)}$ and $\phi^{(5,6)}$.

The utility of $\phi^{(2,4)}$ is defined similarly as $\phi^{(1,3)}$.

\subsection{Comparison of the Reconstructed Images in Evaluation} \label{subsec:mnist_visualization}

\begin{figure}[t]
\centering

\subfigure[Original images.]{
\includegraphics[scale=0.4]{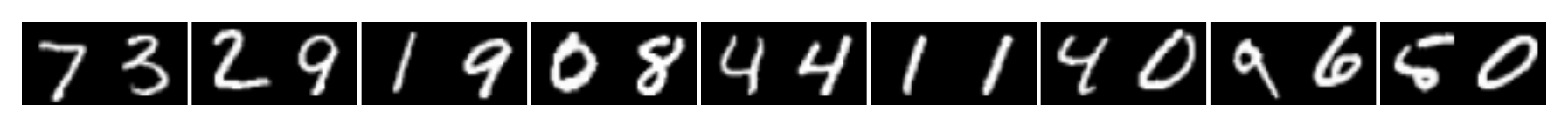}
}
\subfigure[Reconstructed images, task-aware coding.]{
\includegraphics[scale=0.4]{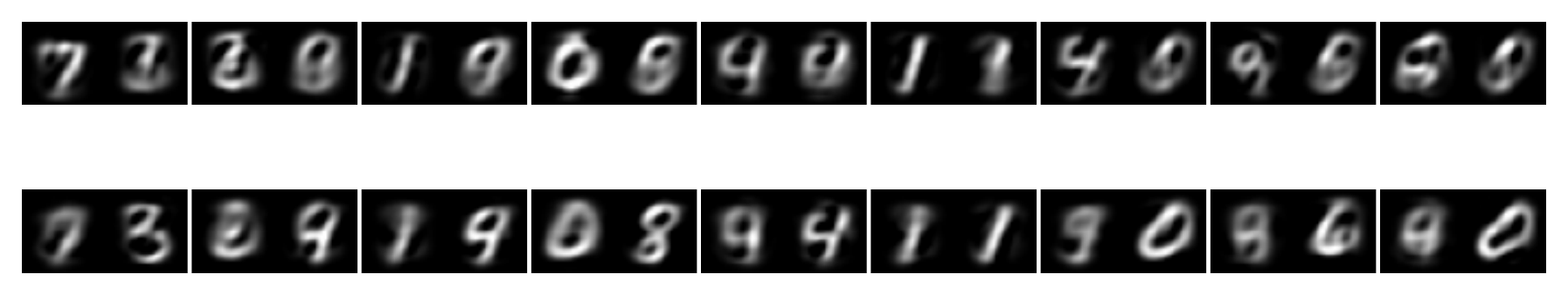}
}
\subfigure[Reconstructed images, task-aware no coding.]{
\includegraphics[scale=0.4]{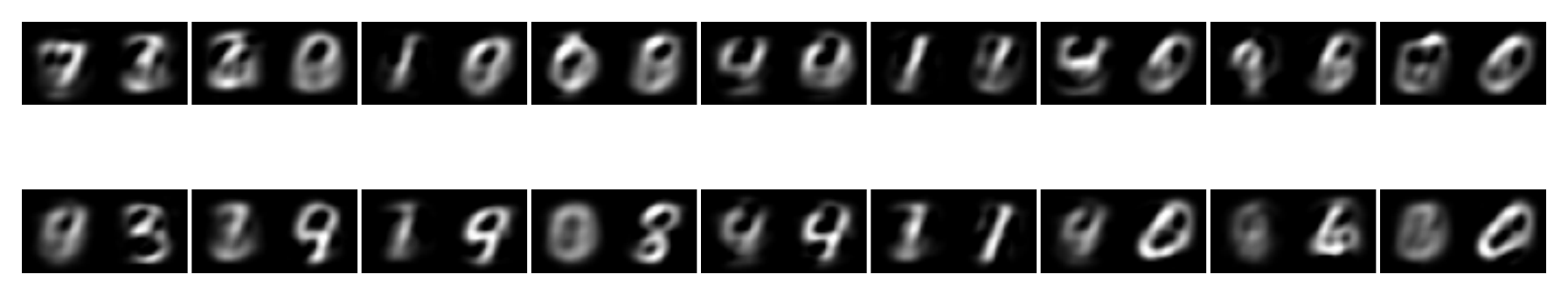}
}
\subfigure[Reconstructed images, task-agnostic coding.]{
\includegraphics[scale=0.4]{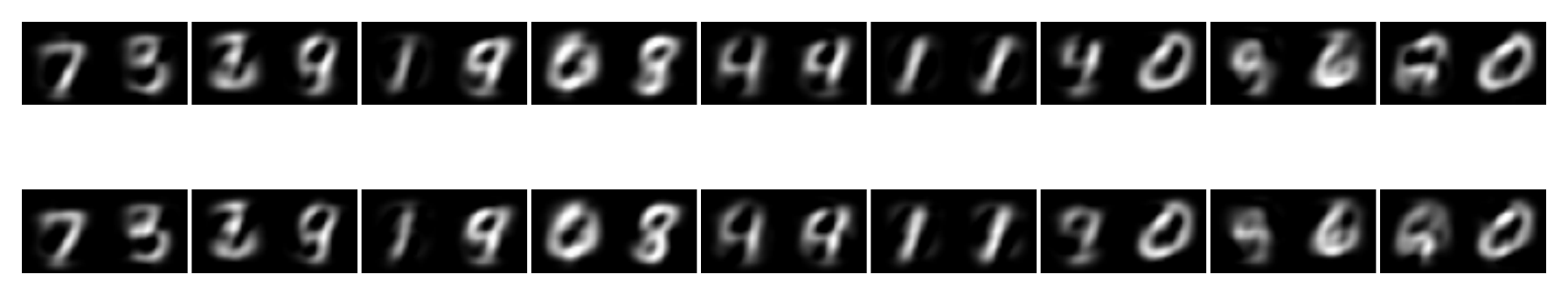}
}
\subfigure[Reconstructed images, coding benchmark.]{
\includegraphics[scale=0.4]{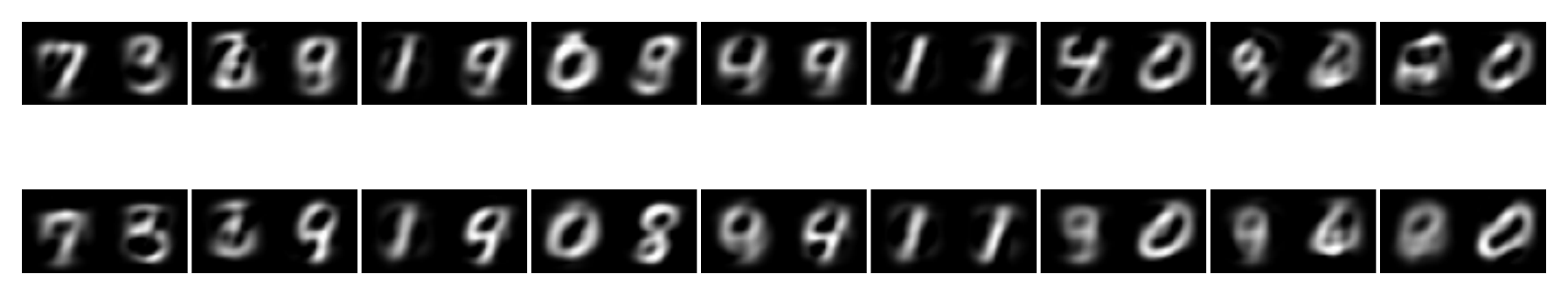}
}
\caption{Original images and reconstructed images for different approaches ($Z=10, \gamma=0.9$). The first and the second row of each approach are the reconstructed images at node 3 and 4, respectively.}
\label{fig:visualization}
\end{figure}

In Fig. \ref{fig:visualization}, we compare a few reconstructed images for different approaches in our evaluation (when $Z=10, \gamma=0.9$). For task-aware approaches, since the task loss of node 3 is dominated by the reconstruction of the left feature map respectively, we see the left part of the reconstructed image at node 3 has a higher quality compared to the right part. And for node 4 we have similar observations. On the other hand, for task-agnostic coding approach, the quality difference between the left part and the right part is not as obvious as the task-aware approaches.

\end{document}